\documentclass{scrartcl}

\usepackage[left=3cm,right=3cm,top=3cm]{geometry}
\usepackage{amsmath,amsfonts,amssymb,amsthm}
\usepackage[utf8]{inputenc}      
\usepackage{tikz}
\usetikzlibrary{decorations.pathreplacing}
\usetikzlibrary{shapes}
\usepackage{subfigure}
\usepackage{bbm}     
\usepackage{hyperref}

\newcommand{\ie}{i.\,e.,}
\newcommand{\eg}{e.\,g.,}
\newcommand{\etal}{et al.}
\newcommand{\REAL}{\mathbbm{R}}
\DeclareMathOperator{\cost}{cost}
\DeclareMathOperator{\ri}{ri}

\DeclareMathOperator{\dphi}{d_{\Phi}}
\DeclareMathOperator{\poly}{poly}

\newcommand{\SinglePnt}{\emph{SinglePnt}}

\newtheorem{theorem}{Theorem}
\newtheorem{fact}[theorem]{Fact}

\newcommand{\NP}{NP} 
\newcommand{\sse}{SSE} 
\newcommand{\sbe}{SBE} 

\renewcommand{\O}{\ensuremath{{\mathcal{O}}}}

\begin{document}

\title{Theoretical Analysis of the \texorpdfstring{$k$}{k}-Means Algorithm -- A Survey}

\author{Johannes Bl\"omer\thanks{Department of Computer Science, University of Paderborn, Germany} \and
        Christiane Lammersen\thanks{School of Computing Science, Simon Fraser University, Burnaby, B.C., Canada} \and 
				Melanie Schmidt\thanks{Computer Science Department, Carnegie Mellon University, Pittsburgh, PA, USA} \and 
				Christian Sohler\thanks{Department of Computer Science, TU Dortmund University, Germany}}

\maketitle


\begin{abstract}
The \texorpdfstring{$k$}{k}-means algorithm is one of the most widely used clustering heuristics. Despite its simplicity, analyzing its running time and 
quality of approximation is surprisingly difficult and can lead to deep insights that can be used to improve the algorithm. In this paper
we survey the recent results in this direction as well as several extension of the basic $k$-means method. 
\end{abstract}


\section{Introduction}
Clustering is a basic process in data analysis. It aims to partition a set of objects into groups called \emph{clusters} such that, ideally, objects in the same group are similar and objects in different groups are dissimilar to each other. 
There are many scenarios where such a partition is useful. 
It may, for example, be used to structure the data to allow efficient information retrieval, to reduce the data by replacing a cluster by one or more representatives or to extract the main \lq themes\rq\ in the data.
There are many surveys on clustering algorithms, including well-known classics~\cite{H75,JMF99} and more recent ones~\cite{B06,J10}. Notice that the title of \cite{J10} is \emph{Data clustering: 50 years beyond K-means} in reference to the \emph{$k$-means algorithm}, the probably most widely used clustering algorithm of all time. It was proposed in 1957 by Lloyd~\cite{L57} (and independently in 1956 by Steinhaus~\cite{S56}) and is the topic of this survey.

The $k$-means algorithm solves the problem of \emph{clustering to minimize the sum of squared errors} (SSE).
In this problem, we are given a set of points $P \subset \REAL^d$ in a Euclidean space, and the goal is to find a set $C \subset \REAL^d$ of $k$ points (not 
necessarily included in $P$) such  that the sum of the squared distances of the points in $P$ to their nearest center in $C$ is minimized. Thus, the objective function 
to be minimized is \[ \cost(P,C) := \sum_{p\in P} \min_{c \in C} \| p - c\|^2 \enspace , \]
where $\|\cdot\|^2$ is the squared Euclidean distance. The points in $C$ are called \emph{centers}.
The objective function may also be viewed as the attempt to minimize the variance of the Euclidean distance of the points to their nearest cluster centers. 
Also notice that when given the centers, the partition of the data set is implicitly defined by assigning each point to its nearest center. 

The above problem formulation assumes that the number of centers $k$ is known in advance. How to choose $k$ might be apparent from the application at hand, or from a statistical model that is assumed to be true. If it is not, then the $k$-means algorithm is typically embedded into a search for the correct number of clusters. It is then  necessary to specify a measure that allows to compare clusterings with different $k$ (the \sse\ criterion is monotonically decreasing with $k$ and thus not a good measure). A good introduction to the topic is the overview by Venkatasubramanian~\cite{V10} as well as Section 5 in the paper by Tibshirani, Walther, and Hastie~\cite{TWH00} and the summary by Gordon \cite{G96}. In this survey, we assume that $k$ is provided with the input.

As Jain~\cite{J10} also notices, the $k$-means algorithm is still widely used for clustering and in particular for solving the SSE problem. That is true despite a variety of alternative options that have been developed in fifty years of research, and even though the $k$-means algorithm has known drawbacks.

In this survey, we review the theoretical analysis that has been developed for the $k$-means algorithm. Our aim is to give an overview on the properties of the $k$-means algorithm and to understand its weaknesses, but also to point out what makes the $k$-means algorithm such an attractive algorithm. In this survey we mainly review theoretical
aspects of the $k$-means algorithm, i.e. focus on the deduction part of the algorithm engineering cycle, but we also discuss some
implementations with focus on scalability for big data.

\subsection{The \texorpdfstring{$k$}{k}-means algorithm}

In order to solve the \sse\ problem heuristically, the $k$-means algorithm starts with an initial candidate solution $\{c_1,\dots,c_k\} \subset \REAL^d$, which can be chosen arbitrarily (often, it is chosen as a random subset of $P$).\label{alg:lloyd} Then, two steps are alternated until convergence: First, for each $c_i$, the algorithm calculates the set $P_i$ of all points in $P$ that are closest to $c_i$ (where ties are broken arbitrarily). Then, for each $1\le i \le k$, it replaces $c_i$ by the mean of $P_i$. Because of this calculation of the \lq means\rq\ of the sets $P_i$, the algorithm is also called \emph{the $k$-means algorithm}.

\enlargethispage{\baselineskip}
\begin{tabbing}
{\scshape The $k$-Means Algorithm}\\
Input: Point set $P \subseteq \REAL^d$ \\
\hspace{1.1cm}  number of centers $k$\\
1. \hspace{0.1cm}\= Choose initial centers $c_1,\dots, c_k$ of from $\REAL^d$\\
2. \>{\bfseries repeat} \\
3. \>\hspace{0.5cm}\= $P_1, \dots, P_k \leftarrow \emptyset$\\ 
4. \>\> {\bfseries for each} $p \in P$  {\bfseries do}\\
5. \>\> \hspace{0.5cm}\= Let $i = \arg \min_{i=1,\dots,k} \|p-c_i\|^2$\\
6. \>\>\> $P_i \leftarrow P_i \cup \{p\}$\\
7. \>\> {\bfseries for} $i=1$ {\bfseries to} $k$ {\bfseries do}\\
8. \>\>\> {\bfseries if} $P_i \not= \emptyset$ {\bfseries then} $c_i = \frac{1}{|P_i|} \sum_{p \in P_i} p$\\
9. \>{\bfseries until} the centers do not change 
\end{tabbing}

The $k$-means algorithm is a local improvement heuristic, because replacing the center of a set $P_i$ by its mean can only improve the solution (see Fact~\ref{fact:cost_of_arbitrary_center_set} below), and then reassigning the points to their closest center in $C$ again only improves the solution.  
The algorithm converges, but the first important question is how many iterations are necessary until an optimal or good solution is found. The second natural question is how good the solution will be when the algorithm stops. We survey upper and lower bounds on running time and quality in Section~\ref{sec:basicanalysis}. 
Since the quality of the computed solution depends significantly on the starting solution, we discuss ways to choose the starting set of centers in a clever way in Section~\ref{sec:seeding}. 
Then, we survey variants of the basic $k$-means algorithm in Section~\ref{sec:variants} and alternatives to the $k$-means algorithm in Section~\ref{sec:alternatives}.
In Section \ref{sec:complexity}, we consider the complexity of the \sse\ problem. Finally, we describe results on 
the $k$-means problem and algorithm for Bregman divergences Section~\ref{sec:bregman}. Bregman divergences have numerous applications and constitute the largest class of dissimilarity measure for which the $k$-means algorithm can be applied. 


\section{Running Time and Quality of the basic \texorpdfstring{$k$}{k}-Means Algorithm}\label{sec:basicanalysis}

In this section, we consider the two main theoretical questions about the $k$-means algorithm: What is its running time, and does it provide a solution of a guaranteed quality? We start with the running time.

\subsection{Analysis of Running Time}\label{runningtime}

The running time of the $k$-means algorithm depends on the number of iterations and on the running time for one iteration. 
While the running time for one iteration is clearly polynomial in $n, d$ and $k$, this is not obvious (and in general not true) for the number of iterations. 
Yet, in practice, it is often observed that the $k$-means algorithm does not significantly improve after a relatively small number of steps. 
Therefore, one often performs only a constant number of steps. It is also common to just stop the algorithm after a given maximum number of iterations, even if it has not converged. The running time analysis thus focuses on two things. First, what the asymptotic running time of one iteration is and how it can be accelerated for benign inputs. Second, whether there is a theoretical explanation on why the algorithm tends to converge fast in practice. 

\subsubsection{Running Time of One Iteration}\label{runningtimeoneiteration}
A straightforward implementation computes $\Theta(nk)$ distances in each iteration in time $\Theta(ndk)$ and runs over the complete input point set. We denote this as the \lq naive\rq\ implementation. Asymptotically, the running time for this is dominated by the number of iterations, which is in general not polynomially bounded in $n$ in the worst case (see next subsection for details). However, in practice, the number of iterations is often manually capped, and the running time of one iteration becomes the important factor. We thus want to mention a few practical improvements. 

The question is whether and how it can be avoided to always compute the distances between all points and centers, even if this does not lead to an asymptotic improvement. Imagine the following pruning rule: Let $c_i$ be a center in the current iteration. Compute the minimum distance $\Delta_i$ between $c_i$ and any other center in time $\Theta(kd)$. Whenever the distance between a point $p$ and $c_i$ is smaller than $\Delta_i/2$, then the closest center to $p$ is $c_i$ and computing the other $k-1$ distances is not necessary. A common observation is that points often stay with the same cluster as in the previous iteration. Thus, check first whether the point is within the safe zone of its old center. More complicated pruning rules take the movement of the points into account. If a point has not moved far compared to the center movements, it keeps its center allocation. Rules like this aim at accelerating the $k$-means algorithm while computing the same clustering as a naïve implementation.
The example pruning rules are from~\cite{JKJ98}.

Accelerating the algorithm can also be done by assigning groups of points together using \emph{sufficient statistics}. Assume that a subset $P'$ of points is assigned to the same center. Then finding this center and later updating it based on the new points can be done by only using three statistics on $P'$. These are the sum of the points (which is a point itself), the sum of the squared lengths of the points (and thus a constant) and the number of points. However, this is only useful if the statistic is already precomputed. For low-dimensional data sets, the precomputation can be done using \emph{kd-trees}. These provide a hierarchical subdivision of a point set. The idea now is to equip each inner node with sufficient statistics on the point set represented by it. When reassigning points to centers, pruning techniques can be used to decide whether all points belonging to an inner node have the same center, or whether it is necessary to proceed to the child nodes to compute the assignment. Different algorithms based on this idea are given in~\cite{ARS98,KMNPSW02,PM99}.
Notice that sufficient statistics are used in other contexts, too, e.g. as a building block of the well-known data stream clustering algorithm BIRCH~\cite{ZRL97}.

There are many ways more that help to accelerate the $k$-means algorithm. For an extensive overview and more pointers to the literature, see~\cite{HD15}.


\subsubsection{Worst-Case Analysis}

Now we we take a closer look at the worst-case number of iterations, starting with (large) general upper bounds and better upper bounds in special cases. 
Then we review results for lower bounds on the number of iterations and thus on the running time of the basic $k$-means algorithm. In the next section, we have a look into work on smoothed analysis for the $k$-means algorithm which gives indications on why the $k$-means algorithm often performs so well in practice.

\paragraph*{Upper Bounds} The worst-case running time to compute a $k$-clustering of $n$ points in $\REAL^d$ by applying the $k$-means algorithm is upper bounded by $\O(ndk \cdot T)$, where $T$ is the number of iterations of the algorithm. It is known that the number of iterations of the algorithm is bounded by the number of partitionings of the input points induced by a Voronoi-diagramm of  $k$ centers. 
This number can be bounded by $\O(n^{dk^2})$ because given a set of $k$ centers, we can move each of the $O(k^2)$ bisectors such that they coincide with $d$ linearly independent points without changing the 
partition.
For the special case of $d=1$ and $k<5$, Dasgupta~\cite{Das03} proved an upper bound of $\O(n)$ iterations. Later, for $d=1$ and any $k$, Har-Peled and Sadri~\cite{HS05} showed an upper bound of $\O(n \Delta^2)$ iterations, where $\Delta$ is the ratio between the diameter and the smallest pairwise distance of the input points.

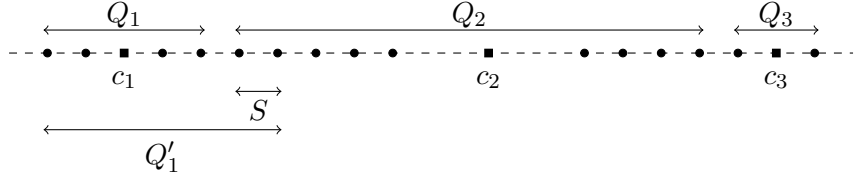
\begin{figure}[ht]
  \centering
  \begin{tikzpicture}[scale=0.51]
    \draw[<->] (0.9,0.6) -- (5.1,0.6);
    \draw[fill=white] (3,1) node{$Q_1$};
    \draw[<->] (5.9,-1) -- (7.1,-1);
    \draw[fill=white] (6.5,-1.5) node{$S$};
    \draw[<->] (0.9,-2) -- (7.1,-2);
    \draw[fill=white] (4,-2.7) node{$Q'_1$};
    \draw[<->] (5.9,0.6) -- (18.1,0.6);
    \draw[fill=white] (12,1) node{$Q_2$};
    \draw[<->] (18.9,0.6) -- (21.1,0.6);
    \draw[fill=white] (20,1) node{$Q_3$};
    
    \draw[style=dashed] (0,0) -- (22,0);
    \draw[fill=black] (1,0) circle (0.1);
    \draw[fill=black] (2,0) circle (0.1);
    \draw[fill=black] (2.9,-0.1) rectangle (3.1,0.1);
    \draw[fill=white] (3,-0.7) node{$c_1$};
    \draw[fill=black] (4,0) circle (0.1);
    \draw[fill=black] (5,0) circle (0.1);
    \draw[fill=black] (6,0) circle (0.1);
    \draw[fill=black] (7,0) circle (0.1);
    \draw[fill=black] (8,0) circle (0.1);
    \draw[fill=black] (9,0) circle (0.1);
    \draw[fill=black] (10,0) circle (0.1);
    \draw[fill=black] (12.4,-0.1) rectangle (12.6,0.1);
    \draw[fill=white] (12.5,-0.7) node{$c_2$};
    \draw[fill=black] (15,0) circle (0.1);
    \draw[fill=black] (16,0) circle (0.1);
    \draw[fill=black] (17,0) circle (0.1);
    \draw[fill=black] (18,0) circle (0.1);
    \draw[fill=black] (19,0) circle (0.1);
    \draw[fill=black] (19.9,-0.1) rectangle (20.1,0.1);
    \draw[fill=white] (20,-0.7) node{$c_3$};
    \draw[fill=black] (21,0) circle (0.1);
  \end{tikzpicture}
  \caption{Illustration of the upper bound for the $k$-means algorithm~\cite{HS05}.}
  \label{fig:upper_bound_in_1D}
\end{figure}

In the following, we will explain the idea to obtain the upper bound given in~\cite{HS05}. The input is a set $P$ of $n$ points with spread $\Delta$ from the Euclidean line $\REAL$. W.l.o.g., we can assume that the minimum pairwise distance in $P$ is~$1$ and the diameter of $P$ is $\Delta$. For any natural number $k$ and for any partition of $P$ into $k$ sets, the clustering cost of $P$ with the means of the subsets as centers is bounded by $\O(n \Delta^2)$. In particular, this holds for the solution of the $k$-means algorithm after the first iteration.
Additionally, the clustering cost of $P$ certainly is $\omega(1)$ as we assumed that the minimum pairwise distance in $P$ is~$1$.
Thus, if we can show that each following iteration decreases the cost by at least some constant amount, then we are done.
Let us now consider the point of time in any iteration of the $k$-means algorithm when the cluster centers have been moved to the means of their respective clusters and the next step is to assign each point to the new closest cluster center.
In this step, there has to be a cluster that is extended or shrunk from its right end.
W.l.o.g. and as illustrated in Figure~\ref{fig:upper_bound_in_1D}, let us assume that the leftmost cluster $Q_1$ is extended from its right end.
Let $S$ be the set of points that join cluster $Q_1$ to obtain cluster $Q'_1$.
Since the minimum pairwise distance is~$1$, the distance of the mean of $S$ to the leftmost point in $S$ is at least $(|S|-1)/2$.
Similarly, the distance of the mean of $Q_1$ to the rightmost point in $Q_1$ is at least $(|Q_1|-1)/2$.
Furthermore, the distance between any point in $Q_1$ and any point in $S$ is at least $1$.
Let $\mu(X)$ be the mean of any point set $X$.
Then, we have $\|\mu(Q_1) - \mu(S)\| \ge (|Q_1|-1)/2 + (|S|-1)/2 + 1 = (|Q_1|+|S|)/2$.
The movement of the mean of the leftmost cluster is at least
\begin{align*}
\| \mu(Q_1) - \mu(Q'_1) \| & = \left\| \mu(Q_1) - \frac{|Q_1| \mu(Q_1) + |S| \mu(S)}{|Q_1| + |S|} \right\|  \\
& = \frac{|S|}{|Q_1|+|S|} \| \mu(Q_1) - \mu(S) \| \ge \frac{|S|}{2} \ge \frac{1}{2} \enspace.
\end{align*}
We will now need the following fact, which is proved in Section \ref{complexity}.
\begin{fact}\label{fact:cost_of_arbitrary_center_set}
Let
\[ \mu := \frac{1}{|P|}\sum_{p\in P} p \]
be the mean of a point set $P$, and let $y \in \REAL^d$ be any point. Then, we have
\[ \sum_{p \in P} \|p-y\|^2 = \sum_{p\in P} \|p-\mu\|^2 + |P| \cdot \|y-\mu\|^2 \enspace . \]
\end{fact}

Due to this fact, the result of the above calculation is an improvement of the clustering cost of at least $1/4$, which shows that in each iteration the cost decreases at least by some constant amount and hence there are at most $\O(n \Delta^2)$ iterations.

\paragraph*{Lower Bounds} Lower bounds on the worst-case running time of the $k$-means algorithm have been studied in~\cite{AV06,Das03,Vat09}. Dasgupta~\cite{Das03} proved that the $k$-means algorithm has a worst-case running time of $\Omega(n)$ iterations. Using a construction in some $\Omega(\sqrt{n})$-dimensional space, Arthur and Vassilvitskii~\cite{AV06} were able to improve this result to obtain a super-polynomial worst-case running time of $2^{\Omega(\sqrt{n})}$ iterations. This has been simplified and further improved by Vattani~\cite{Vat09} who proved an exponential lower bound on the worst-case running time of the $k$-means algorithm showing that $k$-means requires $2^{\Omega(n)}$ iterations even in the plane. A modification of the construction shows that the $k$-means algorithm has a worst-case running time that, besides being exponential in $n$, is also exponential in the spread $\Delta$ of the $d$-dimensional input points for any $d \ge 3$. 

In the following, we will give a high-level view on the construction presented in~\cite{Vat09}.
Vattani uses a special set of $n$ input points in $\REAL^2$ and a set of $k=\Theta(n)$ cluster centers adversarially chosen among the input points.
The points are arranged in a sequence of $t = \Theta(n)$ gadgets $G_0,G_1,\ldots,G_{t-1}$.
Except from some scaling, the gadgets are identical.
Each gadget contains a constant number of points, has two clusters and hence two cluster centers, and can perform two stages reflected by the positions of the two centers.
In one stage, gadget $G_i$, $0 \le i < t$, has one center in a certain position $c^*_i$, and, in the other stage, the same center has left the position $c^*_i$ and has moved a little bit towards gadget $G_{i+1}$.
Once triggered by gadget $G_{i+1}$, $G_i$ performs both of these stages twice in a row.
Performing these two stages happens as follows.
The two centers of gadget $G_{i+1}$ are assigned to the center of gravity of their clusters, which results in some points of $G_{i+1}$ are temporarily assigned to the center $c^*_i$ of $G_i$.
Now, the center of $G_i$ located at $c^*_i$ and the centers of $G_{i+1}$ move, so that the points temporarily assigned to a center of $G_i$ are again assigned to the centers of $G_{i+1}$.
Then, again triggered by $G_{i+1}$, gadget $G_i$ performs the same two stages once more. There is only some small modification in the arrangement of the two clusters of $G_{i+1}$.
Now, assume that all gadgets except $G_{t-1}$ are stable and the centers of $G_{t-1}$ are moved to the centers of gravity of their clusters.
This triggers a chain reaction, in which the gadgets perform $2^{\Omega(t)}$ stages in total.
Since, each stage of a gadget corresponds to one iteration of the $k$-means algorithm, the algorithm needs $2^{\Omega(n)}$ iterations on the set of points contained in the gadgets.


\subsubsection{Smoothed Analysis} Concerning the above facts, one might wonder why $k$-means works so well in practice. To close this gap between theory and practice, the algorithm has also been studied in the model of smoothed analysis~\cite{AMR09,AV09,MR09}. This model is especially useful when both worst-case and average-case analysis are not realistic and reflects the fact that real-world datasets are likely to contain measurement errors or imprecise data. In case an algorithm has a low time complexity in the smoothed setting, it is likely to have a small running time on real-world datasets as well.  

Next, we explain the model in more detail. For given parameters $n$ and $\sigma$, an adversary chooses an input instance of size $n$. Then, each input point is perturbed by adding some small amount of random noise using a Gaussian distribution with mean $0$ and standard deviation $\sigma$. The maximum expected running time of the algorithm executed on the perturbed input points is measured.

Arthur and Vassilvitskii~\cite{AV09} showed that, in the smoothed setting, the number of iterations of the $k$-means algorithm is at most $\mathrm{poly}(n^k,\sigma^{-1})$. This was improved by Manthey and R\"oglin~\cite{MR09} who proved the upper bounds $\mathrm{poly}(n^{\sqrt{k}}, 1/\sigma)$ and $k^{kd} \cdot \mathrm{poly}(n, 1/\sigma)$ on the number of iterations. Finally, Arthur et al.~\cite{AMR09} showed that $k$-means has a polynomial-time smoothed complexity of $\mathrm{poly}(n, 1/\sigma)$.

In the following, we will give a high-level view on the intricate analysis presented in~\cite{AMR09}. Arthur \etal\ show that after the first iteration of $k$-means, the cost of the current clustering is bounded by some polynomial in $n$, $k$ and $d$. In each further iteration, either some cluster centers move to the center of gravity of their clusters or some points are assigned to a closer cluster center or even both events happen. Obviously, the clustering cost is decreased after each iteration, but how big is this improvement? Arthur et al. prove that, in expectation, an iteration of $k$-means decreases the clustering cost by some amount polynomial in $1/n$ and $\sigma$. This results in a polynomial-time smoothed complexity.

The key idea to obtain the above lower bound on the minimum improvement per iteration is as follows. Let us call a configuration of an iteration, defined by a partition into clusters and a set of cluster centers, \emph{good} if in the successive iteration either a cluster center moves significantly or reassigning a point decreases the clustering cost of the point significantly. Otherwise, the configuration is called \emph{bad}. Arthur et al. show an upper bound on the probability that a configuration is bad. The problem is now that there are many possible configurations. So we cannot take the union bound over all of these possible configurations to show that the probability of the occurrence of any bad configuration during a run of $k$-means is small. To avoid this problem, Arthur et al. group all configurations into a small number of subsets and show that each subset contains either only good configurations or only bad configurations. Finally, taking the union bound over all subsets of configurations leads to the desired result, i.e., proving that the occurrence of a bad configuration during a run of $k$-means is small.


\subsection{Analysis of Quality}\label{quality}

As mentioned above, the $k$-means algorithm is a local improvement heuristic. 
It is known that the $k$-means algorithm converges to a local optimum~\cite{-SI84} and that no approximation ratio can be guaranteed~\cite{KMNP+04}.
Kanungo \etal~\cite{KMNP+04} illustrate the latter fact by the simple example given in Figure~\ref{fig:no_approx_factor}.
In this example, we are given four input points on the Euclidean line depicted by the first dashed line in Figure~\ref{fig:no_approx_factor}. The distances between the first and second, second and third and third and fourth point are named $x, y$ and $z$, respectively. We assume that $x < y < z$, so $x$ is the smallest distance and placing two centers in the first two points and one between the third and fourth costs $2 \cdot x^2/4=x^2/2$, and this is the (unique) optimal solution depicted on the second dashed line.

On the third dashed line, we see a solution that is clearly not optimal because it costs $y^2/2$ and $y > x$. The approximation ratio of this solution is $y^2 / x^2$, which can be made arbitrarily bad by moving the first point to the left and thus increasing $y$.

If we choose the initial centers randomly, it can happen that the $k$-means algorithm encounters this solution (for example when we pick the first, third and fourth point as initial centers and keep $y < z$ while increasing $y$). 
When finding the solution, the $k$-means algorithm will terminate because the assignment of points to the three centers is unique and every center is the mean of the points assigned to it.

Thus, the worst-case approximation guarantee of the $k$-means algorithm is unbounded.

\begin{figure}[ht]
  \centering
  \begin{tikzpicture}[scale=0.51]
    \draw[|-|] (1,3.7) -- (6,3.7);
    \draw[fill=white] (3.5,4) node{$y$};
    \draw[-|] (6,3.7) -- (13,3.7);
    \draw[fill=white] (9.5,4) node{$z$};
    \draw[-|] (13,3.7) -- (17,3.7);
    \draw[fill=white] (15,4) node{$x$};
    
    \draw (-3,3) node{input points};
    \draw[style=dashed] (0,3) -- (18,3);
    \draw[fill=black] (1,3) circle (0.2);
    \draw[fill=black] (6,3) circle (0.2);
    \draw[fill=black] (13,3) circle (0.2);
    \draw[fill=black] (17,3) circle (0.2);
    
    \draw (-3,1.5) node{optimal centers};
    \draw[style=dashed] (0,1.5) -- (18,1.5);
    \draw[fill=black] (0.8,1.3) rectangle (1.2,1.7);
    \draw[fill=black] (5.8,1.3) rectangle (6.2,1.7);
    \draw[fill=black] (14.8,1.3) rectangle (15.2,1.7);
    
    \draw (-3,0) node{heuristic centers};
    \draw[style=dashed] (0,0) -- (18,0);
    \draw[fill=gray] (3.3,-0.2) rectangle (3.7,0.2);
    \draw[fill=gray] (12.8,-0.2) rectangle (13.2,0.2);
    \draw[fill=gray] (16.8,-0.2) rectangle (17.2,0.2);
  \end{tikzpicture}
  \caption{Example illustrating the fact that no approximation guarantee can be given for the $k$-means algorithm~\cite{KMNP+04}.}
  \label{fig:no_approx_factor}
\end{figure}
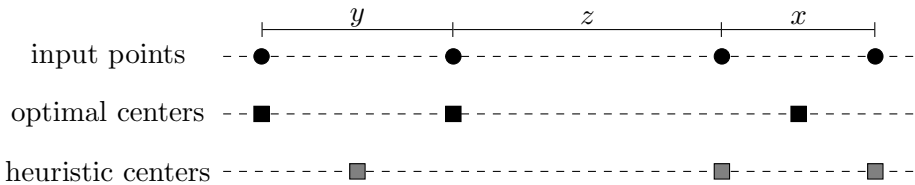


\section{Seeding Methods for the \texorpdfstring{$k$}{k}-means Algorithm}\label{sec:seeding}

The $k$-means algorithm starts with computing an initial solution, which can be done in a number of different ways. Since the $k$-means algorithm is a local improvement strategy we can, in principle, start with an arbitrary solution and then the algorithms runs until it converges to a local optimum. 
However, it is also known that the algorithm is rather sensible to the choice of the starting centers. For example, in the situation in Figure~\ref{fig:no_approx_factor}, no problem occurs if we choose the first, second and third point as the starting centers. 

Often one simply chooses the starting centers uniformly at random, but this can lead to problems, for example, when there is a cluster that is far away from the remaining points and that is so small that it is likely that no point of it is randomly drawn as one of the initial centers.  In such a case one must hope to eventually converge to a solution that has a center in this cluster as otherwise we would end up with a bad solution. Unfortunately, it is not clear that this happens (in fact, one can assume that it will not). 

Therefore, a better idea is to start with a solution that already satisfies some approximation guarantees and let the $k$-means algorithm refine the solution. 
In this section we will present methods that efficiently pick a relatively good initial solution. As discussed later in Section \ref{complexity} there are better approximation algorithms, but they are relatively slow and the algorithms presented in this section present a better trade-off between running time and quality of the initial solution.


\subsection{Adaptive Sampling}\label{adaptivesampling}

Arthur and Vassilvitskii~\cite{AV07} proposed a seeding method for the $k$-means algorithm which applies \emph{adaptive sampling}. They construct an initial set $C$ of $k$ centers in the following way:  
The first center is sampled uniformly at random. For the $i$th center, each input point $p$ is sampled with probability $D^2(p) / \sum_{q \in P}D^2(q)$, where $P$ is the input point set, $D^2(p)=\min_{c_1,\ldots,c_{i-1}} ||p-c_i||^2$ is the cost of $p$ in the current solution and $c_1, \ldots c_{i-1}$ are the centers chosen so far. 
The sampling process is referred to as \emph{$D^2$-sampling}, and the algorithm consisting of $D^2$-sampling followed by the $k$-means algorithm is called \emph{$k$-means}++.


We study the progress of $D^2$-sampling in comparison to a fixed optimal solution. An optimal set of centers partitions $P$ into $k$ \emph{optimal clusters}. If we could sample a center from each cluster uniformly at random, we would in expectation obtain a constant approximation. Since taking a point uniformly at random can also be described as first choosing the cluster and then picking the point uniformly at random, we know that the first point will be uniformly from one (unknown) cluster, which is fine. We want to make sure that this will also approximately be the case for the remaining clusters. The main problem is that there is a significant probability to sample points from a cluster which we already hit (especially, if these clusters contain a lot of points). In order to avoid this, we now sample points with probability proportional to the squared distance from the previously chosen cluster centers. In this way, it is much more likely to sample points from the remaining clusters since the reason that these points belong to a different cluster is that otherwise they would incur a high cost. One can show that in a typical situation, when one of the remaining clusters is far away from the clusters we already hit, then conditioned on the fact that we hit this cluster, the new center will be approximately uniformly distributed within the cluster.  
In the end, this process leads to a set of $k$ centers that is an expected $O(\log k)$-approximation \cite{AV07}. 
 
Thus, $D^2$-sampling is actually an approximation algorithm by itself (albeit one with a worse approximation guarantee than other approximations). It has a running time of $O(kdn)$ and is easy to implement. In addition, it serves well as a seeding method. Arthur and Vassilvitskii obtain experimental results indicating that $k$-means++ outperforms the $k$-means algorithm in practice, both in quality and running time. It also leads to better results than just using $D^2$-sampling as an independent algorithm.

In follow-up work, Aggarwal \etal~\cite{ADK09} show that when sampling $\mathcal{O}(k)$ centers instead of $k$ centers, one obtains a constant-factor approximation algorithm for \sse. This is a bicriteria approximation because in addition to the fact that the clustering cost might not be optimal, the number of clusters is larger than $k$. 


\paragraph*{Adaptive Sampling under Separation Conditions}
Clustering under separation conditions is an interesting research topic on its own. The idea is that the input to a clustering problem should have some structure, otherwise, clustering it would not be meaningful. Separation conditions assume that the optimal clusters cannot have arbitrary close centers or a huge overlap.

We focus on initialization strategies for the $k$-means algorithm. In this paragraph, we will see a result on adaptive sampling that uses a separation condition. In Section~\ref{svdbestfit}, we will see another example for the use of separation conditions. 
Other related work includes the paper by Balcan \etal~\cite{BBG09}, who proposed the idea to recover a \lq true\rq\ (but not necessarily optimal) clustering and introduced assumptions under which this is possible. Their model is stronger than the model by Ostrovsky \etal~\cite{ORSS06} that we will describe next and triggered a lot of follow-up work on other clustering variants.

Ostrovsky \etal~\cite{ORSS06} analyze adaptive sampling under the following $\varepsilon$-separa\-bility: The input is $\varepsilon$-separated if clustering it (optimally) with $k-1$ instead of the desired $k$ clusters increases the cost by a factor of at least $1/\varepsilon^2$. Ostrovsky \etal\ show that under this separation condition, an approach very similar to the above $k$-means++ seedings performs well\footnote{Notice that though we present these results after~\cite{AV07} and~\cite{ADK09} for reasons of presentation, the work of Ostrovsky \etal~\cite{ORSS06} appeared first.}. In their seeding method, the first center is not chosen uniformly at random, but two centers are chosen simultaneously, and the probability for each pair of centers is proportional to their distance. Thus, the seeding starts by picking two centers with rather high distance instead of choosing one center uniformly at random and then picking a center with rather high distance to the first center. 
Ostrovsky \etal\ show that if the input is $\varepsilon$-separated, this seeding achieves a $(1+f(\varepsilon))$-approximation for \sse\ where $f(\varepsilon)$ is a function that goes to zero if $\varepsilon$ does so. The success probability of this algorithm decreases exponentially in $k$ (because there is a constant chance to miss the next cluster in every step), so Ostrovsky \etal\ enhance their algorithm by sampling $\mathcal{O}(k)$ clusters and using a greedy deletion process to reduce the number back to $k$. Thereby, they gain a linear-time constant-factor approximation algorithm (under their separation condition) that can be used as a seeding method.

Later, Awasthi \etal~\cite{ABS10} improved this result by giving an algorithm where the approximation guarantee and the separation condition are decoupled, \ie\ parameterized by different parameters. Braverman \etal~\cite{BMORST11} developed a streaming algorithm.

Note that $\varepsilon$-separability scales with the number of clusters. Imagine $k$ optimal clusters with the same clustering cost $\mathcal{C}$, \ie\ the total clustering cost is $k \cdot \mathcal{C}$. Then, $\varepsilon$-separability requires that clustering with $k-1$ clusters instead of $k$ clusters costs at least $k \cdot \mathcal{C} / \varepsilon^2$. 
Thus, for more clusters, the pairwise separation has to be higher. 


\subsection{Singular Value Decomposition and Best-Fit Subspaces}\label{svdbestfit}

In the remainder of this section, we will review a result from a different line of research because it gives an interesting result for the SSE problem when we make certain input assumptions.

\paragraph*{Learning Mixtures of Gaussians}
In machine learning, clustering is often done from a different perspective, namely as a problem of learning parameters of mixture models. 
In this setting, a set of observations $\mathcal{X}$ is given (in our case, points) together with a statistical model, \ie\ a family of density functions over a set of parameters $\Theta=\{\Theta^1,\ldots,\Theta^\ell\}$. It is assumed that $\mathcal{X}$ was generated by the parameterized density function for one specific parameter set and the goal is to recover these parameters. Thus, the desired output are parameters which explain $\mathcal{X}$ best, \eg\ because they lead to the highest likelihood that $\mathcal{X}$ was drawn.

For us, the special case that the density function is a mixture of Gaussian distributions on $\mathbb{R}^d$ is of special interest because it is very related to \sse. Here, the set of observations $\mathcal{X}$ is a point set which we denote by $P$. On this topic, there has been a lot of research lately, which started by Dasgupta~\cite{-D08} who analyzed the problem under separation conditions. Several improvements were made with separation conditions \cite{AMcS05,AK05,BV08,CR08,DS07,KSV08,VW04} and without separation conditions \cite{BS09,BS10,BS10-2,FOS08,KMV10,MV10}. The main reason why this work cannot be directly applied to \sse\ is the assumption that the input data $\mathcal{X}$ is actually drawn from the parameterized density function so that properties of these distributions can be used and certain extreme examples become unlikely and can be ignored. 
However, in~\cite{KK10}, the authors prove a result which can be decoupled from this assumption, and the paper proposes an initialization method for the $k$-means algorithm. So, we take a closer look at this work.  

Kumar and Kannan~\cite{KK10} assume a given \emph{target clustering} which is to be recovered and then show the following. If $(1-\varepsilon)\cdot |P|$ points in $P$ satisfy a special condition which they call \emph{proximity condition} (which depends on the target clustering), then applying a certain initialization method and afterwards running the $k$-means algorithm leads to a partitioning of the points that \emph{misclassifies} at most $\mathcal{O}(k^2 \varepsilon n)$ points. 
Kumar and Kannan also show that in many scenarios like learning of Gaussian mixtures, points satisfy their proximity condition with high probability. %

Notice that for $\varepsilon = 0$ their result implies that all points are correctly classified, \ie\, the optimal partitioning is found.  
This in particular implies a result for the $k$-means algorithm which is the second step of the algorithm by Kumar and Kannan: It converges to the \lq true\rq\ centers provided that the condition holds for all points. We take a closer look at the separation condition. 


\paragraph*{Separation Condition}
To define the proximity condition, consider the $|P|\times d$ matrix $A$ which has the points of $P$ in its rows. Also define the matrix $C$ by writing the optimal center of the point in row $i$ of $A$ in row $i$ of $C$ (this implies that there are only $k$ different rows vectors in $C$). Now, let $T_1,\ldots,T_k$ be the target clustering, let $\mu_i$ be the mean of $T_i$, and let $n_i$ be the number of points in $T_i$. Then, define
\[
\Delta_{rs} := \left( \frac{ck}{\sqrt{n_r}} + \frac{ck}{\sqrt{n_s}} \right) \|A-C\|_S
\]
for each $r \neq s$ with $r,s \in \{1,\ldots,k\}$, where $c$ is some constant. The term $\|A-C\|_S$ is the \emph{spectral norm} of the matrix $A-C$, defined by
\[
\|A-C\|_S:=\max_{v\in \mathbb{R}^d, \|v\| = 1} \| (A-C) \cdot v \|^2 \enspace .
\]
A point $p$ from cluster $T_r$ satisfies the \emph{proximity condition} if, for any $s\neq r$, the projection of $p$ onto the line between $\mu_r$ and $\mu_s$ is at least $\Delta_{rs}$ closer to $\mu_r$ than to $\mu_s$. 

We have a closer look at the definition.
The term $A-C$ is the matrix consisting of the difference vectors, \ie\ it gives the deviations of the points to their centers. The term $\| (A-C) \cdot v \|^2$ is the projection of these distance vectors into direction $v$, \ie\ a measure on how much the data is scattered in this direction. Thus, $\|A-C\|_S / n $ is the largest average distance to the mean in any direction. It is an upper bound on the variance of the optimal clusters. Assume that $n_i = n/k$ for all $i$. Then, $\Delta_{rs}^2 = (2c)^2 k^2 \|A-C\|_S^2 /n_i$ is close to being the maximal average variance of the two clusters in any direction. It is actually larger, because $\|A-C\|_S$ includes all clusters, so $\Delta_{rs}$ and thus the separation of the points in $T_r$ and $T_s$ depends on all clusters even though it differs for different $r,s$. 


\paragraph*{Seeding Method}
Given an input that is assumed to satisfy the above separation condition, Kumar and Kanan compute an initial solution by projecting the points onto a lower-dimensional subspace and approximately solving the low-dimensional instance. The computed centers form the seed to the $k$-means method. 

The lower-dimensional subspace is the best-fit subspace $V_k$, \ie\ it minimizes the expression $\sum_{p \in P} \min_{v\in V} \|p-v\|^2$ among all $k$-dimensional subspaces $V$. It is known that $V_k$ is the subspace spanned by the first $k$ eigenvectors of $A$, which can be calculated by \emph{singular value decomposition (SVD)}\footnote{The computation of the SVD is a well-studied field of research. For an in-depth introduction to spectral algorithms and singular value decompositions, see~\cite{KV09}. }, and that projecting points to $V_k$ and solving the SSE optimally on the projected points yields a $2$-approximation. Any constant-factor approximation thus gives a constant approximation for the original input.

In addition to these known facts, the result by Kumar and Kannan shows that initializing the $k$-means algorithm with this solution even yields an optimal solution as long as the optimal partition satisfies the proximity condition.

\section{Variants and Extensions of the \texorpdfstring{$k$}{k}-means Algorithm}\label{sec:variants}

The $k$-means algorithm is a widely used algorithm, but not always in the form given above. Naming all possible variations of the algorithm is beyond the scope of this survey and may be impossible to do. We look at two theoretically analyzed modifications.

\paragraph*{Single Point Assignment Step} We call a point in a given clustering \emph{misclassified} if the distance to the cluster center it is currently assigned to is longer than the distance to at least one of the other cluster centers.
Hence, in each iteration of the $k$-means algorithm, all misclassified points are assigned to their closest cluster center and then all cluster centers are moved to the means of the updated clusters.
Har-Peled and Sadri~\cite{HS05} study a variant of the $k$-means algorithm in which the assignment step assigns only one misclassified point to the closest cluster center instead of all misclassified points at once as done in the original algorithm. After such an assignment step, the centers of the two updated clusters are moved to the means of the clusters.
The algorithm repeats this until no misclassified points exist.
Har-Peled and Sadri call their variant \SinglePnt.
Given a number of clusters $k$ and a set $P$ of $n$ points with spread $\Delta$ from a Euclidean space $\REAL^d$, they show that the number of iterations of \SinglePnt\ is upper bounded by some polynomial in $n$, $\Delta$, and $k$.

In the following, we will describe the proof given in~\cite{HS05}.
W.l.o.g., we can assume that the minimum pairwise distance in $P$ is~$1$ and the diameter of $P$ is $\Delta$.
As we have seen for the classical $k$-means algorithm, the cost of $P$ is $\O(n \Delta^2)$ after the first iteration of \SinglePnt.
The main idea is now to show that, in each following iteration of \SinglePnt, the improvement of the clustering cost is lower bounded by some value dependent on the ratio between the distance of the reassigned point to the two involved cluster centers and the size of the two clusters.
Based on this fact, we will prove that $\O(kn)$ iterations are sufficient to decrease the clustering cost by some constant amount, which results in $\O(kn^2\Delta^2)$ iterations in total.

Let $Q_i$ and $Q_j$ be any two clusters such that, in an assignment step, a point $q \in Q_j$ moves from cluster $Q_j$ to cluster $Q_i$, \ie\ after this step we obtain the two clusters $Q'_i = Q_i \cup \{q\}$ and $Q'_j = Q_j \backslash \{ q \}$.
Let $\mu(X)$ be the mean of any point set $X \subset \REAL^d$.
Then, the movement of the first cluster center is
\[ \| \mu(Q_i) - \mu(Q'_i) \| = \left\| \mu(Q_i) - \left( \frac{|Q_i|}{|Q_i|+1}\mu(Q_i) + \frac{1}{|Q_i|+1}q \right) \right\| = \frac{\| \mu(Q_i)-q \|}{|Q_i|+1} \enspace . \]
Similarly, we have $\| \mu(Q_j) - \mu(Q'_j) \| = \| \mu(Q_j)-q \| / (|Q_j|-1)$.
Due to Fact~\ref{fact:cost_of_arbitrary_center_set}, the movement of the first cluster center decreases the clustering cost of $Q'_i$ by $(|Q_i|+1)\| \mu(Q_i) - \mu(Q'_i) \|^2 = \| \mu(Q_i)-q \|/(|Q_i|+1)$, and the movement of the second cluster center decreases the clustering cost of $Q'_j$ by $(|Q_j|-1)\| \mu(Q_j) - \mu(Q'_j) \|^2 = \| \mu(Q_j)-q \|/(|Q_j|-1)$. It follows that the total decrease in the clustering cost is at least $(\| \mu(Q_i)-q \| + \| \mu(Q_j)-q \|)^2/(2(|Q_i|+|Q_j|))$.

The reassignment of a point $q \in P$ is called \emph{good} if the distance of $q$ to at least one of the two centers of the involved clusters is bigger than $1/8$.
Otherwise, the reassignment is called \emph{bad}.
If a reassignment is good, then it follows from the above that the improvement of the clustering cost is at least $(1/8)^2/(2n) = 1/(128n)$.
Thus, $\O(n)$ good reassignments are sufficient to improve the clustering cost by some constant amount.
Next, we show that one out of $k+1$ reassignments must be good, which then completes the proof.

For each $i \in \{ 1,\ldots,k \}$, let $B_i$ be the ball with radius $1/8$ whose center is the $i$-th center in the current clustering.
Since the minimum pairwise distance in $P$ is~$1$, each ball can contain at most one point of $P$.
Observe that a point $q \in P$ can only be involved in a bad reassignment if it is contained in more than one ball.
Let us consider the case that, due to a bad reassignment, a ball $B_i$ loses its point $q \in P$ and so has been moved a distance of at most $1/8$ away from~$q$.
Since the minimum pairwise distance in $P$ is~$1$, $B_i$ needs a good reassignment, so that it can again contain a point from $P$.
Next, observe that, while performing only bad reassignments, a cluster $Q_i$ is changed by gaining or losing the point $q$ contained in $B_i$.
Hence, if a cluster $B_i$ loses $q$, it cannot gain it back. Otherwise, the clustering cost would be increased.
It follows that the total number of consecutive bad reassignments is at most $k$.

\paragraph*{Generalization of Misclassification}
Har-Peled and Sadri~\cite{HS05} study another variant of the $k$-means algorithm, which they call \emph{Lazy-$k$-Means}.
This variant works exactly like the original algorithm except that each iteration reassigns only those points which are significantly misclassified.
More precisely, given a $k$-clustering of a set $P$ of $n$ points from a Euclidean space $\REAL^d$ and a precision parameter $\varepsilon$, $0 \le \varepsilon \le 1$, we call a point $q \in P$ $(1+\varepsilon)$-misclassified if $q$ belongs to some cluster $Q_j$ and there is some other cluster $Q_i$ with $\| q-\mu(Q_j) \| > (1+\varepsilon) \| q-\mu(Q_i) \|$, where $\mu(X)$ is the mean of some set $X \subset \REAL^d$.
Each iteration of \emph{Lazy-$k$-Means} reassigns all $(1+\varepsilon)$-misclassified points to their closest cluster center and then moves each cluster center to the mean of its updated cluster.
This process is repeated until there are no $(1+\varepsilon)$-misclassified points.
Note that, for $\varepsilon=0$, \emph{Lazy-$k$-Means} is equal to the $k$-means algorithm.
For $0 < \varepsilon \le 1$, Har-Peled and Sadri prove that the number of iteration of \emph{Lazy-$k$-Means} is upper bounded by some polynomial in $n$, $\Delta$, and $\varepsilon^{-1}$, where $\Delta$ is the spread of the point set $P$.

In the following, we will sketch the proof given in~\cite{HS05}.
W.l.o.g., we can assume that the minimum pairwise distance in $P$ is~$1$ and the diameter of $P$ is $\Delta$, so the clustering cost is $\O(n \Delta^2)$ after the first iteration of \emph{Lazy-$k$-Means}.
The idea is now to show that every two consecutive iterations lead to a cost improvement of $\Omega(\varepsilon^3)$, which results in $\O(n \Delta^2 \varepsilon^{-3})$ iterations in total.
The proof of the lower bound on the cost improvement is based on the following known fact (see also Figure~\ref{fig:eps-Apollonius-ball}).

\begin{figure}[ht]
  \centering
  \begin{tikzpicture}[scale=0.3]
    \draw[color=black!15!white,fill=black!15!white] (12,0) circle (6);
    \draw[|<->|] (12,0) -- (12,6);
    \draw (15.3,3) node{$R = \frac{\ell(1+\varepsilon)}{\varepsilon(2+\varepsilon)}$};
    \draw[|<->|] (5.5,-2) -- (12,-2);
    \draw (8.75,-3.3) node{$R + \frac{\ell\varepsilon}{2(2+\varepsilon)}$};
    
    \draw (0,0) -- (21,0);
    \draw[fill=black] (3,0) circle (0.2);
    \draw (3,-1) node{$c$};
    \draw[fill=black] (8,0) circle (0.2);
    \draw (8,-1) node{$c'$};
    \draw[|<->|] (3,1) -- (8,1);
    \draw (4.2,2) node{$\ell$};
    
    \draw[style=dashed] (5.5,-6) -- (5.5,6);
  \end{tikzpicture}
  \caption{Illustration of the $\varepsilon$-Apollonius ball for a point $c'$ with respect to a point $c$.}
  \label{fig:eps-Apollonius-ball}
\end{figure}
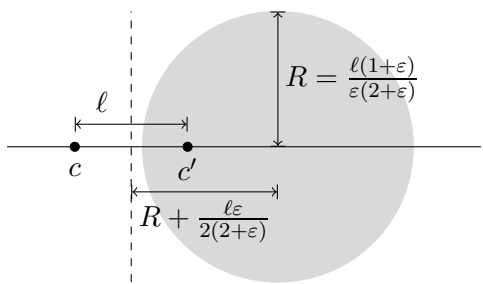

\begin{fact}
\label{fact:eps-Apollonius-ball}
Given two points $c,c' \in \REAL^d$ with $\| c-c' \| = \ell$, all points $q \in \REAL^d$ with $\| q-c \| > (1+\varepsilon)\| q-c' \|$ are contained in the open ball whose radius is $R = \ell(1+\varepsilon)/(\varepsilon(2+\varepsilon))$ and whose center is on the line containing the segment $cc'$ at distance $R+\ell\varepsilon/(2(2+\varepsilon))$ from the bisector of $cc'$ and on the same side of the bisector as $c'$. The ball is called \emph{$\varepsilon$-Apollonius ball} for $c'$ with respect to $c$.
\end{fact}

Let $q \in P$ be any $(1+\varepsilon)$-misclassified point that switches its assignment from a center $c$ to another center $c'$ with $\ell = \| c-c' \|$.
We also say that $c$ and $c'$ are the \emph{switch centers} of $q$.
Then, based on the fact that the distance of $q$ to the bisector of $cc'$ is at least $\ell\varepsilon/(2(2+\varepsilon))$ (see Fact~\ref{fact:eps-Apollonius-ball} and Figure~\ref{fig:eps-Apollonius-ball}) and by using Pythagorean equality, one can show that the improvement of the clustering cost for $q$ is at least
\[ \| q-c \|^2 - \| q-c' \|^2 \ge \frac{\ell^2 \varepsilon}{2+\varepsilon} \enspace . \]
We call any $(1+\varepsilon)$-misclassified point $q \in P$ \emph{strongly misclassified} if the distance between its switch centers is at least $\ell_0 := \varepsilon (2+\varepsilon)/(16(1+\varepsilon))$.
Otherwise, a $(1+\varepsilon)$-misclassified point is called \emph{weakly misclassified}. 
It follows from the above inequality that the improvement of the clustering cost caused by reassigning a strongly misclassified point is at least $\ell_0^2 \varepsilon/(2+\varepsilon) = \Omega(\varepsilon^3)$ for $0 < \varepsilon \le 1$.
Thus, if we can show that at least every second iteration of \emph{Lazy-$k$-Means} reassigns some strongly misclassified point, then we are done.

Let us assume that there are only weakly misclassified points, and $q$ is one of these points with switch centers $c$ and $c'$.
We know that the distance $\ell$ between $c$ and $c'$ is less than $\ell_0$, which is less than $1/8$ for $0 < \varepsilon \le 1$.
Furthermore, it follows from $\ell < \ell_0$ that the radius of the $\varepsilon$-Apollonius ball for $c'$ with respect to $c$ is less than $1/16$ (see also Figure~\ref{fig:proof_Lazy-k-Means}).
Since $q$ is contained in this $\varepsilon$-Apollonius ball, the distance between $c'$ and $q$ is less than $1/8$.
Hence, both switch centers have a distance of less than $1/4$ from $q$.
Since the minimum pairwise distance in $P$ is~$1$, every center can serve as a switch center for at most one weakly misclassified point.

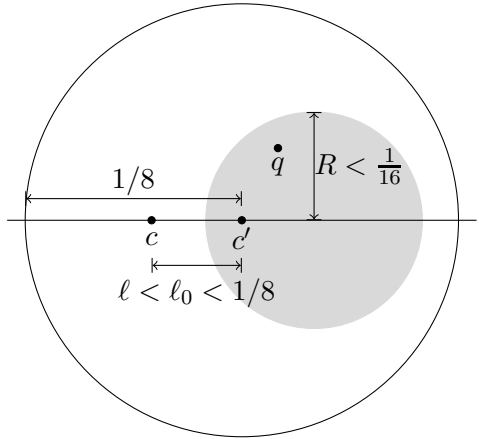
\begin{figure}[ht]
  \centering
  \begin{tikzpicture}[scale=0.24]
    \draw[color=black!15!white,fill=black!15!white] (12,0) circle (6);
    \draw[|<->|] (12,0) -- (12,6);
    \draw (14.5,3) node{$R<\frac{1}{16}$};
    
    \draw (-5,0) -- (21,0);
    \draw[fill=black] (3,0) circle (0.2);
    \draw (3,-1) node{$c$};
    \draw[fill=black] (8,0) circle (0.2);
    \draw (8,-1) node{$c'$};
    \draw[|<->|] (3,-2.5) -- (8,-2.5);
    \draw (5.5,-4) node{$\ell < \ell_0 < 1/8$};
    \draw[|<->|] (-4,1.2) -- (8,1.2);
    \draw (2,2.2) node{$1/8$};
    
    \draw[fill=black] (10,4) circle (0.2);
    \draw (10,3) node{$q$};
    
    \draw (8,0) circle (12);
  \end{tikzpicture}
  \caption{Illustration of the fact that each center can serve as a switch center for at most one weakly misclassified point.}
  \label{fig:proof_Lazy-k-Means}
\end{figure}

Let us consider any weakly misclassified point $q$ with switch centers $c$ and~$c'$, where $c$ belongs to the cluster that loses $q$ and $c'$ belongs to the cluster that gains $q$.
As explained above, both centers have a distance of less than $1/4$ from~$q$.
Hence, due to reassigning $q$, center $c$ is moved by a distance of less than $1/4$.
It follows that, after the considered iteration, the distance between $c$ and $q$ is less than $1/2$.
Since the minimum pairwise distance in $P$ is~$1$, every other point in $P$ has a distance of more than $1/2$ to $c$.
Thus, $c$ can only be a switch center for strongly misclassified points in the next iteration.
Furthermore, due to reassigning $q$, the gaining center $c'$ is moved towards $q$.
Since the distance of $q$ to all the other points in $P$ is at least $1$, no other center can move closer to $q$ than $c'$ due to a reassignment of a weakly misclassified point.
This means in the next iteration $c'$ will still be the closest cluster center to $q$ and $q$ will not be $(1+\varepsilon)$-misclassified.
As a result, either there are no $(1+\varepsilon)$-misclassified points left and the algorithm terminates or there are some strongly misclassified points.
Thus, at least every second iteration reassigns some strongly misclassified points, which  completes the proof.



\section{Alternatives to the \texorpdfstring{$k$}{k}-means algorithm for big data}\label{sec:alternatives}

Again, naming all alternative clustering algorithms that have been proposed is beyond the scope of this survey. However, we will take a short look at algorithms, that are developed 
starting from a theoretical analysis (with respect to the \sse\ problem), but that are also implemented and shown to be viable in practice. We have already discussed one prime example for this type of algorithm, the $k$-means++ algorithm by Arthur and Vassilvitskii~\cite{AV07}. The running time of the seeding is comparable to one iteration of the $k$-means algorithm (when assuming that drawing random numbers is possible in constant time), so using it as a seeding method does not have a significant influence on the running time asymptotically or in practice. However, it turns the $k$-means algorithm into an expected $\mathcal{O}(\log k)$-approximation algorithm. A similar example is the local search algorithm by Kanungo et al.~\cite{KMNPSW02} that we describe in more detail in Section~\ref{approximationalgorithms}. It has a polynomial worst case running time and provides a constant approximation. Additionally, it was implemented and showed very good practical behavior when combined with the $k$-means algorithm.

However, the research we have discussed in Section~\ref{runningtimeoneiteration} aiming at accelerating the iterations of the $k$-means algorithm shows that there is interest in being faster than the $k$-means algorithm (and the constant approximation algorithms), and this interest increases with the availability of larger and larger amounts of data. The problem of solving the \sse\ problem for big data has been researched from a practical as well as from a theoretical side and in this section, we are interested in the intersection. 

The theoretical model of choice is \emph{streaming}. The data stream model assumes that the data can only be read once and in a given order, and that the algorithm is restricted to small space, e.g. polylogarithmic in the input it processes, but still computes an approximation. One-pass algorithms and low memory usage are certainly also desirable from a practical point of view, since random access to the data is a major slowdown for algorithms, and small memory usage might mean that all stored information actually fits into the main memory. 
The $k$-means algorithm reads the complete data set in each iteration, and a straightforward implementation of the $k$-means++ reads the data about $k$ times for the seeding alone, and these are reasons why the algorithms do not scale so well for large inputs.

An old variant of the $k$-means algorithm, proposed independently of Lloyd's work by MacQueen~\cite{-M67}, gives a very fast alternative to the $k$-means algorithm. It processes the data once, assigns each new data point to its closest center and updates this center to be the centroid of the points assigned to it. Thus, it never reassigns points. MacQueen's $k$-means algorithm clearly satisfies the first two requirements for a streaming algorithm, but not the third. Indeed, it is not surprising that MacQueen's algorithm does not necessarily converge to a good solution, and that the solution depends heavily on the start centers and the order of the input points. 
The famous streaming algorithm BIRCH~\cite{ZRL97} is also very fast and is perceived as producing better clusterings, yet, it still shares the property that there is no approximation guarantee~\cite{FGSSS13}.

Various data stream algorithms for the \sse\ problem have been proposed, see for example~\cite{-C09,FL11,FMS07,FS05,HK07,HM04}, achieving $(1+\varepsilon)$-approximations in one pass over the data for constant $k$ (and constant $d$, for some of the algorithms). We now look at algorithms which lie in between practical and theoretical results.

\subsubsection*{Local search and the Stream framework}

Guha et al.~\cite{GMMMC03} develop a framework for clustering algorithms in the data stream setting that they call \emph{Stream}. They combine it with a constant factor approximation based on local search. The resulting algorithm is named StreamLS\footnote{\url{http://infolab.stanford.edu/~loc/}}. It computes a constant approximation in the data stream setting. 
StreamLS has originally been designed for the variant of the \sse\ problem where the distances are not squared (also called the $k$-median problem), but it is stated to work for the \sse\ problem as well with worse constants. 

The Stream framework reads data in blocks of size $m$. For each block, it computes a set of $c \cdot k$ centers that are a constant factor approximation for the \sse\ problem with $k$ centers ($c$ is a constant) by using an approximation algorithm $A$. It thus reduces $m$ points to $c \cdot k$ points, where $m$ is at least $n^\varepsilon$ for some $\varepsilon > 0$. This is repeated until the number of computed centers reaches $m$, i.e. it is repeated for $m/(ck)$ blocks. Then, $m^2/(ck)$ points have been be reduced to $m$ points. These are then again reduced to $c k$ points, i.e. the computed centers are treated like as input to the same procedure, one level higher in a computation tree. On the $i$th level of this tree, $ck$ points represent $(m/ck)^i$ input blocks. Thus, the height of the computation tree is at most $\mathcal{O}(\log_{m/(ck)} n/m) \in \mathcal{O}(\log_{n^\varepsilon} n)$. This is actually a constant, since 
\[
\log_{n^\varepsilon/(ck)} n = \frac{\log n}{\log n^\varepsilon} = \frac{1}{\varepsilon}.
\]
Thus, the computation tree has constant height. It stores at most $m$ points on each level, so the storage requirement of the algorithm is $\Theta(m) = \mathcal{O}(n^\varepsilon)$ under the assumption that $A$ requires space that is linear in its input size. The running time of the algorithm is $\mathcal{O}(ndk)$ under the assumption that $A$ has linear running time. Whenever an actual solution to the \sse\ problem is queried, it can be produced from the $\mathcal{O}(m)$ stored centers by computing a constant factor approximation by a different algorithm $A'$. Guha et al.\ show that the result is a constant factor approximation for the original input data.

Guha et al.\ also develop the algorithm LSEARCH which they use as the algorithm $A$ within their framework. The algorithm StreamLS is the combination of the Stream framework with the algorithm LSEARCH. LSEARCH is a local search based algorithm that is based on algorithms for a related problem, the facility location problem. It is allowed to computed more than $k$ centers, but additional centers are penalized. The main purpose of LSEARCH is an expected speed-up compared to other local search based methods with $\mathcal{O}(n^2)$ running time.

The experiments included in~\cite{GMMMC03} actually use the \sse\ criterion to evaluate their results, since the intention is to compare with the $k$-means algorithm, which is optimized for \sse. The data sets are around fifty thousand points and forty dimensions.
First, LSEARCH is compared to the $k$-means algorithm and found to be about three times slower than the $k$-means algorithm while producing results that are much better. Then, StreamLS is compared to BIRCH and to StreamKM, the algorithm resulting from embedding the $k$-means algorithm into the Stream framework. StreamLS and StreamKM compute solutions of much higher quality than BIRCH, with StreamLS computing the best solutions. BIRCH on the other hand is significantly faster, in particular, its running time per input point increases much less with increasing stream length. 

\subsubsection*{Adaptions of \texorpdfstring{$k$}{k}-means++}

Ailon, Jaiswal and Monteleoni~\cite{AJM09} use the Stream framework and combine it with different approximation algorithms. The main idea is to extend the seeding part of the $k$-means++ algorithm to an algorithm called $k$-means\# and to use this algorithm within the above Stream framework description. Recall that the seeding in $k$-means++ is done by $D^2$-sampling. This method iteratively samples $k$ centers. The first one is sampled uniformly at random. For the $i$th center, each input point $p$ is sampled with probability $D^2(p) / \sum_{q \in P}D^2(q)$, where $P$ is the input point set, $D^2(p)=\min_{c_1,\ldots,c_{i-1}} ||p-c_i||^2$ is the cost of $p$ in the current solution and $c_1, \ldots c_{i-1}$ are the centers chosen so far. A set of $k$ centers chosen in this way is an expected $\mathcal{O}(\log k)$-approximation. 

The algorithm $k$-means\# starts with choosing $3 \log k$ centers uniformly at random and then performs $k-1$ iterations, each of which samples $3 \log k$ centers according to the above given probability distribution. This is done to ensure that for an arbitrary optimal clustering of the points, each of the clusters is \lq hit\rq\ with constant probability by at least one center. Ailon et al.\ show that the $\mathcal{O}(k \log k)$ centers computed by $k$-means\# are a constant factor approximation for the \sse\ criterion with high probability\footnote{As briefly discussed in Section~\ref{adaptivesampling}, it is sufficient to sample $\mathcal{O}(k)$ centers to obtain a constant factor approximation as later discovered by Aggarwal et al~\cite{ADK09}.}. 

To obtain the final algorithm, the Stream framework is used. Recall that the framework uses two approximation algorithms $A$ and $A'$. While $A$ can be a bicriteria approximation that computes a constant factor approximation with $c\cdot k$ centers, $A'$ has to compute an approximative solution with $k$ centers. The approximation guarantee of the final algorithm is the guarantee provided by $A'$.

Ailon et al.\ sample $k$ centers by $D^2$-sampling for $A'$, thus, the overall result is an expected $\mathcal{O}(\log k)$ approximation. For $A$, $k$-means\# is ran $3 \log n$ times to reduce the error probability sufficiently and then the best clustering is reported. The overall algorithm needs $n^\varepsilon$ memory for a constant $\varepsilon > 0$.

The overall algorithm is compared to the $k$-means algorithm and to MacQueen's $k$-means algorithm on data sets with up to ten thousand points in up to sixty dimensions.
While it produces solutions of better quality than the two $k$-means versions, it is slower than both. 

Ackermann et al. \cite{AMRSLS12} develop a streaming algorithm based on $k$-means++ motivated from a different line of work\footnote{Project website online can be found at \url{http://www.cs.uni-paderborn.de/fachgebiete/ag-bloemer/forschung/abgeschlossene/clustering-dfg-schwerpunktprogramm-1307/streamkmpp.html}.}. The ingredients of their algorithms look very much alike the basic building blocks of the algorithm by Ailon et al.: sampling more than $k$ points according to the $k$-means++ sampling method, organizing the computations in a binary tree and computing the final clustering with $k$-means++. There are key differences, though.

Firstly, their work is motivated from the point of view of \emph{coresets} for the \sse\ problem. A coreset $S$ for a point set $P$ is a smaller and weighted set of points that has approximately the same clustering cost as $P$ \emph{for any choice of $k$ centers}. It thus satisfies a very strong property. Ackermann et al.\ show that sampling sufficiently many points according to the $k$-means++ sampling results in a coreset. For constant dimension $d$, they show that $\mathcal{O}( k \cdot (\log n)^{O(1)})$ points guarantee that the clustering cost of the sampled points is within an $\varepsilon$-error from the true cost of $P$ for any set of $k$ centers\footnote{This holds with constant probability and for any constant $\varepsilon$.}. 

Coresets can be embedded into a streaming setting very nicely by using a technique called \emph{merge-and-reduce}. It works similar as the computation tree of the Stream framework: It reads blocks of data, computes a coreset for each block and merges and reduces these coresets in a binary computation tree. Now the advantage is that this tree can have superconstant height since this can be cancelled out by adjusting the error $\varepsilon$ of each coreset computation. A maximum height of $\Theta(\log n)$ means that the block size on the lowest level can be much smaller than above (recall that in the algorithm by Ailon et al., the block size was $n^{\varepsilon}$). For the above algorithm, a height of $\Theta(\log n)$ would mean that the approximation ratio would be $\Omega(c^{\log n})\in\Omega(n)$.
By embedding their coreset construction into the merge-and-reduce technique, Ackermann et al.\ provide a streaming algorithm that needs $\mathcal{O}(k \cdot (\log n)^{O(1)})$ space and computes a coreset of similar size for \sse\ problem. They obtain a solution for the problem by running $k$-means++ on the coreset. Thus, the solution is an expected $\mathcal{O}(\log k)$-approximation. 

Secondly, Ackermann et al.\ significantly speed up the $k$-means++ sampling approach. Since the sampling is applied again and again, this has a major impact on the running time. Notice that it is necessary for the sampling to compute $D(p)$ for all $p$ and to update this after each center that was drawn. When computing a coreset of $m$ points for a point of $\ell$ points, a vanilla implementation of this sampling needs $\Theta(d m \ell)$ time. Ackermann et. al.\ develop a data structure called \emph{coreset tree} which allows to perform the sampling much faster. It does, however, change the sampling procedure slightly, such that the theoretically proven bound does not necessarily hold any more. 

In the actual implementation, the sample size and thus the coreset size is set to $200k$ and thus much smaller than it is supported by the theoretical analysis. However, experiments support that the algorithm still produces solutions of high quality, despite these two heuristic changes.
The resulting algorithm is called StreamKM++.

Ackermann et al.\ test their algorithm on data sets with up to eleven million points in up to $68$ dimensions and compare the performance to BIRCH, StreamLS, the $k$-means algorithm and $k$-means++. They find that StreamLS and StreamKM++ compute solutions of comparable quality, and much better than BIRCH. BIRCH is the fastest algorithm. However, StreamKM++ beats the running time of StreamLS by far and can e.g. compute a solution for the largest data set and $k=30$ in $27 \%$ of the running time of StreamLS. For small dimensions or higher $k$, the speed up is even larger. The $k$-means algorithm and $k$-means++ are much slower than StreamLS and thus also than StreamKM++. It is to be expected that StreamKM++ is faster than the variant by Ailon et al.\ as well. 

\subsection*{Sufficient statistics}
The renown algorithm BIRCH\footnote{http://pages.cs.wisc.edu/~vganti/birchcode/}~\cite{ZRL97} computes a clustering in one pass over the data by maintaining a preclustering. It uses a data structure called \emph{clustering feature} tree, where the term clustering feature denotes the sufficient statistics for the \sse\ problem. The leaves of the tree represent subsets of the input data by their sufficient statistics. At the arrival of each new point, BIRCH decides whether to add the point to an existing subset or not. If so, then it applies a rule to choose one of the subsets and to add the point to it by updating the sufficient statistics. This can be done in constant time. If not, then the tree grows and represents a partitioning with one more subset. 

BIRCH has a parameter for the maximum size of the tree. If the size of the tree exceeds this threshold, then it rebuilds the tree. Notice that a subset represented by its sufficient statistics cannot be split up. Thus, rebuilding means that some subsets are merged to obtain a smaller tree. After reading the input data, BIRCH represents each subset in the partitioning by a weighted point (which is obtained from the sufficient statistics) and then runs a clustering algorithm on the weighted point set.

The algorithm is very fast since updating the sufficient statistics is highly efficient and rebuilding does not occur too often. However, the solutions computed by BIRCH are not guaranteed to have a low cost with respect to the \sse\ cost function. 

Fichtenberger et al.\ \cite{FGSSS13} develop the algorithm BICO\footnote{\texttt{http://ls2-www.cs.uni-dortmund.de/bico}}. The name is a combination of the words BIRCH and coreset. BICO also maintains a tree which stores a representation of a partitioning. Each node of this tree represents a subset by its sufficient statistics.

The idea of BICO is to improve the decision if and where to add a point to a subset in order to decrease the error of the summary. For this, BICO maintains a maximum error value $T$. A subset is forbidden to induce more error than $T$. The error of a subset is measured by the squared distances of all points in the subset to the centroid because in the end of the computation, the subset will be represented by the centroid.

For a new point, BICO searches for the subset whose centroid is closest to the point. BICO first checks whether the new point lies within a certain radius of this centroid since it wants to avoid to use all the allowed error of a subset for one point. If the point lies outside of the radius, a new node is created directly beneath the root of the tree for the new point.  Otherwise, the point is added to this subset if the error keeps being bounded by $T$. If the point does not pass this check, then it is passed on to the child node of the current node whose centroid is closest. If no child node exists or the point lies without the nodes radius, then a new child node is created based on the new point. 

If the tree gets too large, then $T$ is doubled and the tree is rebuilt by merging subsets whose error as a combined subset is below the new $T$.

For constant dimension $d$, Fichtenberger et al.\ show that the altered method is guaranteed to compute a summary that satisfies the coreset property for a threshold value that lies in $\Theta(k \cdot \log n)$. Combined with $k$-means++, BICO gives an expected $\mathcal{O}(\log k)$-approximation.

The implementation of BICO faces the same challenge as StreamKM++, $k$-means or $k$-means++, namely, it needs to again and again compute the distance between a point and its closest neighbor in a stored point set. BICO has one advantage, though, since it is only interested in this neighbor if it lies within a certain radius of the new point. This helps in developing heuristics to speed up the insertion process. The method implemented in BICO has the same worst case behavior as iterating through all stored points but can be much faster.

Fichtenberger et al.\ compare BICO to StreamKM++, BIRCH and MacQueen's $k$-means algorithm on the same data sets as in~\cite{AMRSLS12} and one additional $128$-dimensional data set. In all experiments, the summary size of BICO is set to $200k$, thus the summary is not necessarily a coreset. The findings are that BICO and StreamKM++ compute the best solutions, while BIRCH and MacQueen are the fastest algorithms. However, for small $k$, the running time of BICO is comparable to BIRCH and MacQueen. The running time of BICO is $\mathcal{O}(ndm)$, where $m$ is the chosen summary size, thus, the increase in the running time for larger $k$ stems from the choice $m=200k$. For larger $k$, the running time can be decreased to lie below the running time of BIRCH by reducing $m$ at the cost of worse solutions. In the tested instances, the quality was then still higher than for BIRCH and MacQueen.


\section{Complexity of \sse}\label{complexity}\label{sec:complexity}


Before we consider variants of the $k$-means algorithm that deal with objective functions different from SSE, we conclude our SSE related study by looking at the complexity of SSE in general. We start by delivering a proof to the following fact which we already used above. We also reflect on the insights that it gives us on the structure of optimal solutions of the SSE problem.

\begin{fact}\label{fact:magic}
Let
$ \mu := \frac{1}{|P|}\sum_{p\in P} p $
be the mean of a point set $P$, and let $y \in \REAL^d$ be any point. Then, we have
\[ \sum_{p \in P} \|p-y\|^2 = \sum_{p\in P} \|p-\mu\|^2 + |P| \cdot \|y-\mu\|^2 \enspace . \]
\end{fact}

\begin{proof}
The result is well known and the proof is contained in many papers. We in particular follow~\cite{KMNP+04}.
First note that
\begin{eqnarray*}
  \sum_{p \in P} \|p-y\|^2
  & =
  & \sum_{p \in P} \|p-\mu+\mu-y\|^2
  \\
  & =
  & \sum_{p \in P} \|p-\mu\|^2 + 2 (\mu-y)^T \sum_{p \in P} (p - \mu) + |P| \cdot \|y-\mu\|^2 \enspace .
\end{eqnarray*}
Thus, the statement follows from
\[ \sum_{p \in P} (p-\mu) = \sum_{p \in P}p - |P| \cdot \mu = \sum_{p \in P}p - |P| \frac{1}{|P|}\sum_{p\in P} p = 0 \enspace . \] 
\end{proof}

The first consequence of Fact~\ref{fact:magic} is that the SSE problem can be solved analytically for $k=1$:
The mean $\mu$ minimizes the cost function, and the optimal cost is $ \sum_{p\in P} \|p-\mu\|^2$.
For $k\ge 2$, the optimal solution induces a partitioning of the input point set $P$ into subsets of P with the same closest center. These subsets are called \emph{clusters}. The center of a cluster is the mean of the points contained in the cluster (otherwise, exchanging the center by the mean would improve the solution).
At the same time, every partitioning of the point set induces a feasible solution by computing the mean of each subset of the partitioning. This gives a new representation of an optimal solution as a partitioning of the input point set that minimizes the induced clustering cost. 

Notice that we cannot easily enumerate all possible centers as there are infinitely many possibilities. By our new view on optimal solutions, we can instead iterate over all possible partitionings. However, the number of possible partitionings is exponential in $n$ for every constant $k \ge 2$. We get the intuition that the problem is hard, probably even for small $k$. Next, we see a proof that this is indeed the case. Notice that there exist different proofs for the fact that \sse\ is NP-hard ~\cite{ADHP09,Das08,-M09} and the proof presented here is the one due to Aloise et al.~\cite{ADHP09}.


\paragraph*{NP-Hardness of \sse}
We reduce the following problem to \sse\ with $k=2$. 
Given a graph $G=(V,E)$, a \emph{cut} is a partitioning of the nodes $V$ into subsets $X \subset V$ and $V\backslash X \subset V$.  
By the \emph{density of a cut} $(X,V\backslash X)$, we mean the ratio $|E(X)|/(|X| \cdot |V \backslash X|)$, where $E(X)$ is the set of edges having one endpoint in $X$ and the other endpoint in $V \backslash X$. Now, our version of the \emph{densest cut problem} asks for the cut with the highest density. This problem is \NP-hard because it is equivalent to finding the cut with minimal density in the complement graph, which is known to be \NP-hard due to~\cite{MS90}.

We define a type of incidence matrix $M$ in the following way. In a $|V|\times |E|$-matrix, the entry in row $i$ and column $j$ is 0 if edge $j$ is not incident to vertex $i$. Otherwise, let $i'$ be the other vertex to which $j$ is incident. Then, we arbitrarily set one of the two entries $(i,j)$ and $(i',j)$ to $1$ and the other one to $-1$. For an example, see Figure~\ref{fig:ex:graph} and \ref{fig:ex:matrix}.
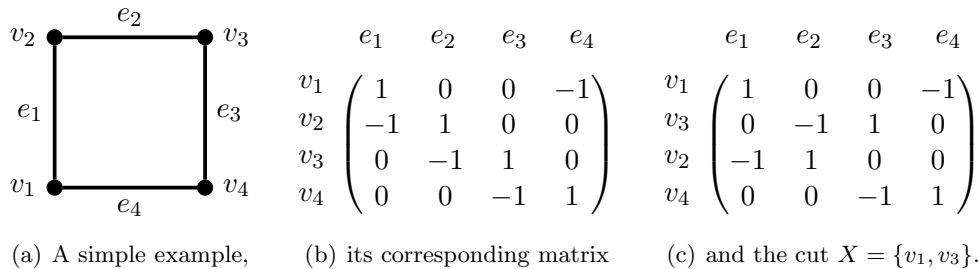
\begin{figure}
\centering
\subfigure[A simple example,\label{fig:ex:graph}]{
\begin{tikzpicture}
\node (a) [shape=circle,inner sep = 0 cm, minimum size = 0.2cm,draw,fill=black,label=left:{$v_1$}] at (0,0) {};
\node (b) [shape=circle,inner sep = 0 cm, minimum size = 0.2cm,draw,fill=black,label=left:{$v_2$}] at (0,2) {};
\node (c) [shape=circle,inner sep = 0 cm, minimum size = 0.2cm,draw,fill=black,label=right:{$v_3$}] at (2,2) {};
\node (d) [shape=circle,inner sep = 0 cm, minimum size = 0.2cm,draw,fill=black,label=right:{$v_4$}] at (2,0) {};
\draw [line width = 0.05cm] (a) -- (b) node[midway,left] {$e_1$};
\draw [line width = 0.05cm] (b) -- (c) node[midway,above] {$e_2$};
\draw [line width = 0.05cm] (c) -- (d) node[midway,right] {$e_3$};
\draw [line width = 0.05cm] (d) -- (a) node[midway,below] {$e_4$};
\end{tikzpicture}
}
\subfigure[its corresponding matrix\label{fig:ex:matrix}]{
\begin{tikzpicture}
\node (e1) at (-1.4,1.4) {$e_1$};
\node (e2) at (-0.45,1.4) {$e_2$};
\node (e3) at (+0.5,1.4) {$e_3$};
\node (e4) at (+1.4,1.4) {$e_4$};
\node (e1) at (-2.2,0.8) {$v_1$};
\node (e2) at (-2.2,0.3) {$v_2$};
\node (e3) at (-2.2,-0.2) {$v_3$};
\node (e4) at (-2.2,-0.7) {$v_4$};
\node (m) at (0,0) 	{$\begin{pmatrix}
	1 & 0 & 0 & -1 \\
	-1 & 1 & 0 & 0 \\
	0 & -1 & 1 & 0 \\
	0 & 0 & -1 & 1
	\end{pmatrix}
	$};
\end{tikzpicture}
}
\subfigure[and the cut $X=\{v_1,v_3\}$.\label{fig:ex:cut}]{
\begin{tikzpicture}
\node (e1) at (-1.4,1.4) {$e_1$};
\node (e2) at (-0.45,1.4) {$e_2$};
\node (e3) at (+0.5,1.4) {$e_3$};
\node (e4) at (+1.4,1.4) {$e_4$};
\node (e1) at (-2.2,0.8) {$v_1$};
\node (e2) at (-2.2,-0.2) {$v_2$};
\node (e3) at (-2.2,0.3) {$v_3$};
\node (e4) at (-2.2,-0.7) {$v_4$};
\node (m) at (0,0) 	{$\begin{pmatrix}
	1 & 0 & 0 & -1 \\
	0 & -1 & 1 & 0 \\
	-1 & 1 & 0 & 0 \\
	0 & 0 & -1 & 1 \\
	\end{pmatrix}
	$};
\end{tikzpicture}
}\caption{An example for the reduction from our densest cut problem to \sse.}
\end{figure}
We interpret the rows of $M$ as points in $\REAL^{|E|}$ and name the set of these points $P(V)$.
Each subset $X \subseteq V$ then corresponds to a subset $P(X)\subseteq P(V)$, and 
a cut $(X,V\backslash X)$ corresponds to a partitioning $(P(X),P(X \backslash V))$ of these points and thus to a $2$-clustering. We take a closer look at the cost of cluster $P(X)$ which is the sum of the costs of all points in it. For each point, the cost is the squared distance to the mean of $P(X)$, and this cost can be calculated by summing up the squared differences in each coordinate. Remember that the coordinates correspond to edges in $E$. Thus, one way to analyze the cost is to figure out how much cost is caused by a specific edge.
For each edge $e_j=(x,y)$, there are three possibilities for the clustering cost: 
If $x, y \in X$, then the mean of $P(X)$ has a $0$ in the $j$th coordinate, and thus the squared distance is $0$ for all coordinates except those corresponding to $x$ and $y$, and it is $1$ for these two. If $x,y \notin X$, then the mean of $P(X)$ also has a $0$ in the $j$th coordinate, and as all points in $P(X)$ also have $0$ at the $j$th coordinate, this coordinate contributes nothing to the total cost. If either $x\in X, y \notin X$ or $x \notin X, y \in X$ and thus $e_j \in E(X)$, then the mean has $\pm 1 / |X|$ as its $j$th coordinate, which induces a squared distance of $(0-1 / |X|)^2$ for $|X|-1$ of the points, and a squared distance of $(1-1 / |X|)^2$ for the one endpoint that is in $X$. Thus, the total cost of $P(X)$ is

\begin{align*}
 & \sum_{e_j=(x,y)\in E, x,y \in X} 2 + \sum_{e_j=(x,y) \in E(X)} \left[ (|X|-1) \frac{1}{|X|^2} + (1-1 / |X|)^2 \right] \\
 = & \sum_{e_j=(x,y)\in E, x,y \in X} 2 + |E(X)|\left( 1 - \frac{1}{|X|} \right) \enspace .
\end{align*}

This analysis holds for the clustering cost of $P( V \backslash X)$ analogously. Additionally, every edge is either in $E(X)$, or it has both endpoints in either $P(X)$ or $P(V \backslash V)$. Thus, the total cost of the 2-clustering induced by $X$ is

\begin{align*}
2 (|E|-|E(X)|) + |E(X)| \left( 2 - \frac{1}{|X|} - \frac{1}{|V\backslash X|} \right) = 2|E| - \frac{|E(X)| \cdot |V|}{|X| \cdot |V \backslash X|} \enspace .
\end{align*}

Finding the optimal 2-clustering means that we minimize the above term. As $2|E|$ and $|V|$ are the same for all possible 2-clusterings, this corresponds to finding the clustering which maximizes $|E(X)| / (|X| \cdot |V \backslash X| )$. Thus, finding the best 2-clustering is equivalent to maximizing the density. 

Notice that the above transformation produces inputs which are $|E|$-dimensional. Thus, \sse\ is hard for constant $k$ and arbitrary dimension. It is also hard for constant dimension $d$ and arbitrary $k$ \cite{-M09}. For small dimension and a small number of clusters $k$, the problem can be solved in polynomial time by the algorithm of Inaba et al.~\cite{IKI94}. 

\paragraph{Approximation Algorithms}\label{approximationalgorithms}

This section is devoted to the existence of approximation algorithms for \sse. First, we convince ourselves that there is indeed hope for approximation algorithms with polynomial running time even if $k$ or $d$ is large. Above, we stated that we cannot solve the problem by enumerating all possible centers as there are infinitely many of them. But what if we choose centers only from the input point set? This does not lead to an optimal solution: Consider $k=1$ and a point set lying on the boundary of a circle. Then the optimal solution is inside the circle (possibly its center) and is definitely not in the point set. However, the solution cannot be arbitrarily bad. Let $k=1$ and let $c \in P$ be a point $p \in P$ which minimizes $\|p-\mu\|^2$, \ie\ it is the point closest to the optimal center (breaking ties arbitrarily). Then, 
\[\begin{array}{rcl}
\cost(P,\{c\}) = \sum_{p\in P} \| p - c\|^2 & \stackrel{\text{Fact \ref{fact:cost_of_arbitrary_center_set}}}{=} & \sum_{p\in P} \left(\| p - \mu\|^2 + \|c-\mu\|^2\right) \\
& \le & \sum_{p\in P} \left(\| p - \mu\|^2 + \|p-\mu\|^2\right) = 2 \cost(P,\{\mu\}) \enspace . 
\end{array}\]
Thus, a 2-approximated solution to the 1-means problem can be found in quadratic time by iterating through all input points. For $k>1$, the calculation holds for each cluster in the optimal solution, and thus there exists a 2-approximate solution consisting of $k$ input points. By iterating through all $\O(n^k)$ possible ways to choose $k$ points from $P$, this gives a polynomial-time approximation algorithm for constant $k$.

For arbitrary $k$, we need a better way to explore the search space, \ie\ the possible choices of centers out of $P$ to gain a constant-factor approximation algorithm with polynomial running time. Kanungo et al.~\cite{KMNP+04} show that a simple \emph{swapping} algorithm suffices. Consider a candidate solution, \ie\ a set $C\subseteq P$ with $|C|=k$. The swapping algorithm repeatedly searches for points $c \in C$ and $p \in P \backslash C$ with $\cost( P,C ) > \cost( P,C \cup \{p\} \backslash \{c\})$, and then replaces $c$ by $p$. Kanungo et al. prove that if no such swapping pair is found, then the solution is a $25$-approximation of the best possible choice of centers from $P$. Thus, the swapping algorithm converges to a $50$-approximation\footnote{Note that Kanungo et al. use a better candidate set and thus give a $(25+\varepsilon)$-approximation.}. In addition, they show that in polynomial time by always taking swaps that significantly improve the solution, one only loses a $(1+\varepsilon)$-factor in the approximation guarantee. This gives a very simple local search algorithm with constant approximation guarantee. Kanungo et al. also refine their algorithm in two ways: First, they use a result by Matou\v{s}ek \cite{-M00} that says that one can find a set $S$ of size $\O(n \varepsilon^{-d} \log(1/\varepsilon))$ in time $\O(n \log n + n \varepsilon^{-d} \log (1 / \varepsilon))$ such that the best choice of centers from $S$ is a $(1+\varepsilon)$-approximation of the best choice of centers from $\REAL^d$. This set is used to choose the centers from instead of simply using $P$. Second, they use \emph{$q$-swaps} instead of the $1$-swaps described before. Here, $q' \le q$ centers are simultaneously replaced by a set of $q'$ new centers. They show that this leads to a $(9 + \varepsilon)$-approximation and also give a tight example showing that $9$ is the best possible approximation ratio for swapping-based algorithms.

The work of Kanungo et al. is one step in a series of papers developing approximation algorithms for \sse. The first constant approximation algorithm was given by Jain and Vazirani \cite{JV01} who developed a primal dual approximation algorithm for a related problem and extended it to the \sse\ setting. Inaba et al.~\cite{IKI94} developed the first polynomial-time $(1+\varepsilon)$-approximation algorithm for the case of $k=2$ clusters. Matu\v{s}ek \cite{-M00} improved this and obtained a polynomial-time $(1+\varepsilon)$-approximation algorithm for constant $k$ and $d$ with running time $\O(n \log^k n)$ if $\varepsilon$ is also fixed. Further $(1+\varepsilon)$-approximations were for example given by \cite{-C09,dlVKKR03,FL11,FS05,HM04,KSS10}. Notice that all cited $(1+\varepsilon)$-approximation algorithms are exponential in the number of clusters $k$ and in some cases additionally in the dimension $d$. 

\paragraph{Inapproximability results}
Algorithms with a $(1+\varepsilon)$-guarantee are only known for the case that $k$ is a constant (and $\varepsilon$ has to be a constant, too). Recently, Awasthi, Charikar, Krishnaswamy and Sinop~\cite{ACKS15} showed that there exists an $\varepsilon$ such that it is NP-hard to approximate \sse\ within a factor of $(1+\varepsilon)$ for arbitrary $k$ and $d$. Their proof holds for a very small value of $\varepsilon$, and a larger inapproximability result is not yet known. 

\section{\texorpdfstring{$k$}{k}-means with Bregman divergences}\label{sec:bregman}
The $k$-means problem can be defined for any dissimilarity measure. An important class of dissimilarity measures are \emph{Bregman divergences}. Bregman divergences have numerous applications in machine learning, data compression, speech and image analysis, data mining, or pattern recognition. We review mainly
results known for the $k$-means algorithm when applied to Bregman divergences. As we will see, for Bregman divergences the $k$-means method can be applied almost without modifications to the algorithm. 

To define Bregman divergences, let $\mathbb{D}\subseteq \REAL^d$, and let 
$\Phi:\mathbb{D}\to \REAL$ be a strictly convex function that is differentiable on the relative interior
$\text{ri}(\mathbb{D})$. The \emph{ Bregman divergence 
$ d_{\Phi}:\mathbb{D}\times \ri(\mathbb{D})\to \REAL_{\ge 0}\cup \{\infty\}$} is defined as 
\begin{align*}
d_{\Phi}(x,c) = \Phi(x)-\Phi(c)-(x-c)^T\nabla\Phi(c),
\end{align*}
where $\nabla\Phi(c)$ is the gradient of $\Phi$ at $c$.
The squared Euclidean distance is a Bregman divergence. Other Bregman divergences that are
used in various applications are shown on Table~\ref{table:bregman-list}.

\begin{figure*}
 \centering
 \small

  \begin{tabular}{|c|l|l|} \hline
   domain $\mathbb{D}$ & $\Phi(x)$ & $d_{\Phi}(p,q)$  \\ \hline \hline
   & squared $\ell_2$-norm &  squared Euclidean distance \\
   $\REAL^d$ & $\|x\|_2^2$ &  $\|x-c\|_2^2$  \\ \hline
   & generalized norm & Mahalanobis distance \\
   $\REAL^d$ & $x^T A x$ &  $(x-c)^T A (x-c)$  \\ \hline
   & neg. Shannon entropy & Kullback-Leibler divergence \\
   $ [0,1]^d$ & $\sum x_i \ln(x_i) $ & $\sum c_i \ln(\frac{c_i}{x_i})$ \\ \hline
   & Burg entropy &   Itakura-Saito divergence \\
   $ \REAL_+^d$ & $\sum - \ln(x_i)$ & $\sum \frac{c_i}{x_i} - \ln(\frac{c_i}{x_i}) - 1$ \\ \hline
   & harmonic $(\alpha > 0)$ & harmonic divergence $(\alpha > 0)$ \\
   $\REAL_+^d$ & $\sum \frac1{x_i^\alpha}$ & $\sum \frac1{c_i^\alpha}-\frac{\alpha+1}{x_i^\alpha} + \frac{\alpha c_i}{x_i^{\alpha+1}}$ \\ \hline
   & norm-like $(\alpha \geq 2)$ & norm-like divergence $(\alpha \geq 2)$ \\
   $ \REAL_+^d$ & $\sum x_i^\alpha$ & $\sum c_i+(\alpha-1)x_i^\alpha + \alpha c_ix_i^{\alpha-1}$ \\ \hline
   & exponential & Exponential loss \\
   $\REAL^d$ & $\sum e^{x_i}$ & $\sum e^{c_i} - (c_i-x_i+1)e^{x_i}$ \\ \hline
   & Hellinger-like &  Hellinger-like divergence \\
   $ (-1,1)^d$ & $\sum -\sqrt{1-x_i^2}$ &  $\sum \frac{1-c_ix_i}{\sqrt{1-x_i^2}} - \sqrt{1-c_i^2}$ \\ \hline
  \end{tabular}
 \caption{Some  Bregman divergences.}
 \label{table:bregman-list}
\end{figure*}

Bregman divergences have a simple geometric interpretation that is shown in Figure~\ref{fig-bregman}.
For $c$ fixed, let $f_c: \REAL^d\to \REAL$ be defined by $f_c(x):\Phi(c)+(x-c)^T\nabla\Phi (c)$. The function $f_c$ is a linear approximation to $\Phi$ at point $c$. Then $d_{\Phi}(x,c)$ is the difference between the true function value $\Phi(x)$ and the value $f_c(x)$ of the linear approximation to $\Phi$ at $c$.  Bregman divergences usually are asymmetric and violate the triangle inequality.  In fact, the only symmetric Bregman divergences are the Mahalanobis divergences (see Table~\ref{table:bregman-list}).
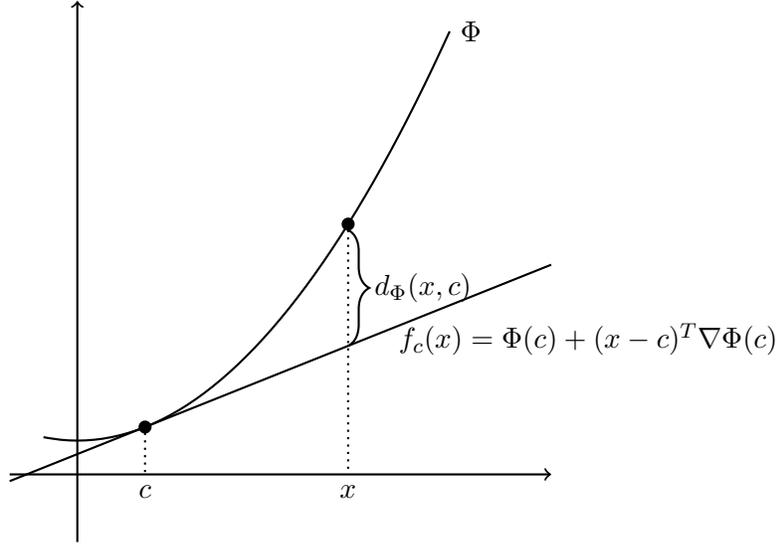
\begin{figure}[!ht]
\begin{center}
\begin{tikzpicture}[scale = .9, style = thick]
	\draw[->] (-1,0)--(7,0) node[right] { };
	\draw[->] (0,-1) -- (0,7) node[above]{ };
	\draw[domain=-.5:5.5,smooth,variable=\x] plot ({\x},{0.2*\x*\x+0.5}) node[right] {$\Phi$};
	\draw[domain=-1:7,smooth,variable=\x] plot ({\x},{ 0.4*\x +0.3}) node[below=18pt, xshift=14pt] 
	{$f_c(x)=\Phi(c)+(x-c)^T\nabla \Phi(c)$};
	\node[draw, circle, style = fill, scale = .4] (c) at (1,.7) { };
	\node[draw, circle, style = fill, scale = .4] (x) at (4,3.7) { };
	
	\draw[-, dotted] (x) --(4,0);
	\draw[-, dotted] (c) --(1,0);
	\draw [decorate,decoration={brace,amplitude=8pt}] (x)  -- (4,1.9)node[midway, right=6pt]{$d_{\Phi}(x,c)$};
	\node[below] at (1,0) {$c$};
	\node[below] at (4,0) {$x$};
\end{tikzpicture}
\caption{Geometric interpretation of Bregman divergences}\label{fig-bregman}
\end{center}
\end{figure}

As one can see from Table~\ref{table:bregman-list}, for some Bregman divergences $d_{\Phi}$ there exist points $x,c$ such that 
$d_{\Phi}(x,c)=\infty$. We call these pairs of points \emph{singularities}.  In most results and algorithms that we describe these singularities require special treatment or have to be defined away.

\paragraph{$k$-means with Bregman divergences.} 
Similar to \sse\ we can define the \emph{minimum sum-of-Bregman-errors clustering problem} (SBE).  In this problem we are given a fixed Bregman divergence $d_{\Phi}$ with domain $\mathbb{D}$ and a set of points $P \subset \mathbb{D}.$ The aim is to find a set $C \subset \ri(\mathbb{D})$ of $k$ points (not necessarily included in $P$) such that the sum of the 
Bregman divergences  of the points in P to their nearest center in $C$ is minimized. Thus, the cost function to be minimized is
\[ \cost_{\Phi}(P,C) := \sum_{p\in P} \min_{c \in C} d_{\Phi}(p,c)\enspace. \]
The points in $C$ are called \emph{centers}. Because of the (possible) asymmetry of $d_{\Phi}$ the order of arguments in $d_{\Phi}(x,c)$ is important.

For any Bregman divergence the optimal solution for $k=1$ is given by the mean of the points in $P$. More precisely, Fact \ref{fact:cost_of_arbitrary_center_set} completely carries over to Bregman divergences (see~\cite{BMDG05}).

\begin{fact}\label{lem:bregman} 
Let $d_{\Phi}:\mathbb{D}\times \ri(\mathbb{D})\to \REAL_{\ge 0}\cup \{\infty\}$ be a Bregman divergence and $P\subset \mathbb{D}, |P|<\infty$ and let
$$
\mu=\frac{1}{|P|}\sum_{p\in P}
$$
be the mean of set $P$. For any $y\in \ri(\mathbb{D})$: 
$$
\sum_{p\in P}d_{\Phi}(p,y)=\sum_{p\in P} d_{\Phi}(p,\mu)+|P|\cdot d_{\Phi}(\mu,y).
$$
\end{fact}
\begin{proof}
It suffices to show the final statement of the Fact. 
\begin{eqnarray*}
\sum_{p\in P}d_{\Phi}(p,y) & = &  \sum_{p\in P} \Phi(p)-\Phi(y)-(x-s)^T\nabla\Phi(y)\\
 & = &  \sum_{p\in P} \Phi(p)-\Phi(\mu)+\Phi(\mu)-\Phi(y)-(x-s)^T\nabla\Phi(y)\\
 & = & \sum_{p\in P}( \Phi(p)-\Phi(\mu))+
|P|(\Phi(\mu)-\Phi(y))-\left(\sum_{p\in P} (p -y)\right)^T\nabla\Phi(y)\\
 & = & \sum_{p\in P}( \Phi(p)-\Phi(\mu))+|P|\bigl(\Phi(\mu)-\Phi(y)-(\mu-y)^T\nabla\Phi(y)\bigr)\\
 & = & \sum_{p\in P} d_{\Phi}(p,\mu)+|P|\cdot d_{\Phi}(\mu,y),
\end{eqnarray*}
where the last equality follows from 
$$
\sum_{p\in P} (p-\mu)^T= 0 \quad\text{and}\quad \sum_{p\in P} (p-\mu)^T\nabla\Phi(\mu)=0.
$$
\end{proof}
Moreover, for all Bregman divergences, any set of input points $P$, and any set of $k$ centers $\{\mu_1,\ldots,\mu_k\}$, the optimal partitions for \sbe\ induced by the centers $\mu_j$ can be separated by hyperplanes. This was first explicitly stated  in~\cite{BMDG05}. More precisely,
the Bregman bisector $\bigl\{x\in \mathbb{D}\mid d_{\Phi}(x,c_1)=d_{\Phi}(x,c_2) \bigr\}$
between any two points $c_1,c_2\in \mathbb{D}\subseteq \REAL^d$ is always a hyperplane. 
i.e.\ for any pair of points $c_1,c_2$  there are $a\in \REAL^d,b\in \REAL$ such that

\begin{align}\label{eq:linearly-separable}
\bigl\{x\in \mathbb{D}\mid d_{\Phi}(x,c_1)=d_{\Phi}(x,c_2) \bigr\}=
\bigl\{x\in \mathbb{D}\mid a^Tx=b\bigr\}.
\end{align}

As a consequence, 
\sbe\ can be solved for any Bregman divergence in time $O(n^{k^2d})$. Hence for fixed $k$ and $d$, \sbe\ is solvable in polynomial time. However, in general \sbe\ is an \NP-hard problem. This was first observed in~\cite{ABS11} and can be shown in two steps. First, let the Bregman divergence $\dphi$ be a Mahalanobis divergence for a
symmetric, positive definite matrix $A$. Then there is a unique symmetric, positive definite matrix
$B$ such that $A=B^TB$, i.e.\ for any $p,q$
\begin{align}\label{eq:mahalanobis}
\dphi(p,q)=(p-q)^TA(p-q)=\|Bp-Bq\|^2.
\end{align}
Therefore, \sbe\ with $\dphi$ is just \sse\ for a linearly transformed input set. This immediately implies that for Mahalanobis divergences \sbe\ is \NP-hard. Next, if $\Phi$ is sufficiently smooth, 
the Hessian $\nabla^2\Phi{t}$ of $\Phi$ at point $t\in \ri(\mathbb{D})$ is a symmetric, positive definite matrix. Therefore, $\dphi$ locally behaves like a Mahalanobis divergence. This can used to show that with appropriate restriction on the strictly convex function $\Phi$ \sbe\ is \NP-hard.

\paragraph{Approximation Algorithms and $\mu$-similarity.} No provable approximation algorithms for general Bregman divergences are known. Approximation algorithms either work for specific Bregman divergences or for restricted classes of Bregman divergences. Chaudhuri and McGregor~\cite{CM08} give an $\O(\log(n))$ approximation algorithm for the Kullback-Leibler divergence ($n$ is the size of the input set $P$). They obtain this result by exploiting relationships between the Kullback-Leibler divergence  and the so-called Hellinger distortion and between the Hellinger distortion and
the squared Euclidean distance. 

The largest subclass of Bregman divergences for which approximation algorithms are known to exist consists of \emph{$\mu$-similar} Bregman divergences. A Bregman divergence $\dphi$ defined on domain $\mathbb{D}\times \ri(\mathbb{D})$  is called $\mu$-similar if there is a symmetric, positive definite matrix $A$ and a constant $0<\mu\le 1$ such that for all $(x,y)\in \mathbb{D}\times \ri(\mathbb{D})$
\begin{align}\label{eq:mu-similar}
\mu \cdot d_A(x,y)\le \dphi(x,y)\le d_A(x,y).
\end{align}
Some Bregman divergences are (trivially) $\mu$-similar. Others, like the Kullback-Leibler divergence or the Itakura-Saito divergence become $\mu$-similar if one restricts the domain on which they are defined. For example, if we restrict the Kullback-Leibler divergence to $\mathbb{D}=[\lambda,\nu]^d$
for $0<\lambda<\nu\le 1$, then the Kullback-Leibler divergence is $\frac{\lambda}{\nu}$-similar. This can be shown by looking at the first order Taylor series expansion of the negative Shannon entropy
$\Phi(x_1,\ldots,x_d)=\sum x_i\ln (x_i)$. 

$\mu$-similar Bregman divergences approximately  behave like Mahalanobis divergences. Due to (\ref{eq:mahalanobis}) Mahalanobis divergences behave like the squared Euclidean distance. Hence, one can hope that $\mu$-similar Bregman divergences behave
roughly like the squared Euclidean distance. In fact, it is not too difficult to show that the swapping algorithm of Kanungo et al.~\cite{KMNP+04} can be generalized to $\mu$-similar Bregman divergences to obtain approximation algorithms with approximation factor $18/\mu^2+\epsilon$ for arbitrary $\epsilon>0$.  Whether one can combine the technique of Kanungo et al.\ with 
Matou\v{s}ek's technique~\cite{-M00} to obtain better constant  factor approximation algorithms is not known.

In the work of Ackermann et al.~\cite{ABS09}, $\mu$-similarity has  been used to obtain a probabilistic $(1+\epsilon)$-approximation algorithm for \sbe, whose running time is exponential in $k,d, 1/\epsilon$, and $1/\mu$, but linear in $|P|$.  Building upon results in~\cite{KSS10}, Ackermann at al.\ describe and analyze an algorithm
to solve the $k$-median problem for metric and non-metric distance measures $D$ that satisfy the following conditions.
\begin{enumerate}
\item[(1)] For $k=1$, optimal solutions to the $k$-median problem with respect to distance $D$ can  be computed efficiently. 
\item[(2)] For every $\delta,\gamma >0$ there is a constant $m_{\delta,\gamma}$ such that for any set $P$, with probability $1-\delta$ the optimal $1$-median of a random sample $S$ of size $m_{\delta,\gamma}$ from $P$ is a $(1+\gamma)$-approximation to the $1$-median for set $P$.
\end{enumerate}
Together, (1) and (2) are called the $[\gamma,\delta]$-sampling property. Using the same algorithm as in ~\cite{KSS10} but a combinatorial rather than geometric analysis, Ackermann et al.\ show that for any
distance measure $D$ satisfying the $[\gamma,\delta]$-sampling property and any $\epsilon >0$ there is
an algorithm that with constant probability returns a $(1+\epsilon)$-approximation to the $k$-median problem with distance measure $D$. The running time of the algorithm is linear in $n$, the number of input points, and exponential in $k,1/\epsilon, $ and the parameter $m_{\delta, \epsilon/3}$ from the sampling property. Finally, Ackermann et al.\ show that any $\mu$-similar Bregman divergence
satisfies the $[\delta,\gamma]$-sampling property with parameter 
$m_{\delta,\gamma}=\frac{1}{\gamma\delta\mu}$. Overall, this yields a $(1+\epsilon)$ algorithm for 
\sbe\  for $\mu$-similar Bregman divergences with running time linear in $n$, and exponential in 
$k,1/\epsilon, 1/\mu$.

\paragraph{The $k$-means algorithm for Bregman divergences.} The starting point for much of the recent research on \sbe\ for Bregman divergences is the work by Banerjee et al.~\cite{BMDG05}. They were the first to explicitly state Fact~\ref{lem:bregman} and describe the $k$-means algorithm (see page~\pageref{alg:lloyd})  as a generic algorithm to solve \sbe\ for arbitrary Bregman divergences.  Surprisingly, the $k$-means algorithm cannot be generalized beyond Bregman divergences. In~\cite{BGW05} it is shown, that under some mild smoothness conditions, any divergence that satisfies Fact~\ref{lem:bregman} is a Bregman divergence. Of course, this does not imply that variants or modifications of the $k$-means algorithm cannot be used  for distance measures other than Bregman divergences. However, in these generalizations cluster centroids cannot be used as optimizers in the second step, the re-estimation step. 

Banerjee et al.\  already  showed that for any Bregman divergence the $k$-means algorithm terminates after a finite number of steps. In fact, using the linear separability of intermediate solutions computed by the $k$-means algorithm (see Eq.~\ref{eq:linearly-separable}), for any Bregman divergence the number of iterations of the $k$-means algorithm can be bounded by 
$\O(n^{k^2d})$. 
Since the squared Euclidean distance is a Bregman divergence it is clear that no approximation guarantees can be given for the solutions the $k$-means algorithm finds for \sbe.  

%

\paragraph{1.\ Lower bounds.} Manthey and R\"oglin extended Vattani's exponential lower bound for the running time of the $k$-means algorithm  to any Bregman divergence $\dphi$ defined by a sufficiently smooth function $\Phi$. In their proof they use an approach similar to the approach used by Ackerman et al.\ to show  that \sbe\ is \NP-hard. Using 
(\ref{eq:mahalanobis}) Manthey and R\"oglin first extend Vattani's lower bound to any Mahalanobis divergence. Then, using the fact that any Bregman divergence $\dphi$ with sufficiently smooth $\Phi$ locally resembles some Mahalanobis divergence $d_{A}$, Manthey and R\"oglin show that a  lower bound for the  Mahalanobis divergence $d_{A}$ carries over to a lower bound for the Bregman divergence $\dphi$. Hence, for any smooth Bregman divergence the $k$-means algorithm has exponential running time. Moreover,  Manthey and R\"oglin show that for the $k$-means algorithm the squared Euclidean distance, and more generally Mahalanobis divergences, are the easiest Bregman divergences.

\paragraph{2.\ Smoothed analysis.} Recall that the smoothed complexity of the $k$-means algorithm is polynomial in $n$ and $1/\sigma$, when each input point is perturbed by random noise generated using a Gaussian distribution with mean $0$ and standard deviation $\sigma$, a result due to Arthur, Manthey, and R\"oglin~\cite{AMR09}. So far, this result has not been generalized to Bregman divergences. 
For almost any Bregman divergence $\dphi$ Manthey and R\"oglin~\cite{MR13} prove two upper bounds
on the smoothed complexity of the $k$-means algorithm. The first bound is of the form 
$\poly(n^{\sqrt{k}}, 1/\sigma)$, the second is of the form $k^{kd}\cdot \poly(n,1/\sigma)$. These bounds match bounds that Manthey and R\"ogin achieved for the squared Euclidean distance in~\cite{MR09}. Instead of reviewing their proofs, we will briefly review two technical 
difficulties Manthey and R\"oglin had to account for.

Bregman divergences $\dphi:\mathbb{D}\times \ri(\mathbb{D})\to \REAL_{\ge 0}\cup\{\infty\}$ like the Kullback-Leibler divergence are defined on a bounded subset of some $\REAL^d$. Therefore perturbing a point in $\mathbb{D}$ may yield a point for which the Bregman divergence is not defined. Moreover, whereas the Gaussian noise is natural for the squared Euclidean distance this is by no means clear for all Bregman divergences. In fact, Banerjee et al.~\cite{BMDG05} already showed a close connection between Bregman divergences and exponential families, indicating that noise chosen according to an exponential distribution may be appropriate for some Bregman divergences. Manthey and R\"oglin deal with these issues by first introducing a general and abstract perturbation model parametrized by some $\sigma\in (0,1]$. Then Manthey and R\"oglin give a smoothed analysis of the $k$-means algorithm for Bregman divergences with respect to this abstract model. It is important to note that as in the squared Euclidean case, the parameter $\sigma$ measures the amount of randomness in the perturbation. Finally, for Bregman divergences like the Mahalanobis divergences, the Kullback-Leibler divergence, or the Itakura-Saito Manthey and R\"oglin instantiate the abstract perturbation model with some  perturbations schemes using explicit distributions.

Singularities of Bregman divergences are the second technical difficulty that 
Manthey and R\"oglin have to deal with. For each Bregman divergence $\dphi$ they introduce two parameters $0<\zeta\le 1 $ and $\xi\ge 1$ that in some sense measures how far away $\dphi$ is from being a Mahalanobis divergence. This resembles the $\mu$-similarity introduced by Ackermann et al.~\cite{ABS09}. Whereas for many Bregman divergences the parameter $\mu$ can only be defined by restricting the domain of the divergence, this is not necessary in the approach by Manthey and 
R\"ogin. However, their upper bounds on the smoothed complexity of the $k$-means algorithm for Bregman divergences are not uniform, instead for any specific Bregman divergence the bound depends (polynomially) on the values $\xi$ and $1/\zeta$. 

It is still an open problem whether the polynomial bound of Arthur et al.~\cite{AMR09} on the smoothed complexity of the $k$-means algorithm can be generalized to Bregman divergences. Surprisingly, even for general Mahalanobis divergences this is not known. As Manthey and R\"oglin mention, at this point polynomial bounds on the smoothed complexity of the $k$-means algorithm can only be achieved for Mahalanobis divergences $d_A$ and input sets $P$, where the largest eigenvalue of $A$ is bounded by a polynomial in $|P|$.

\paragraph{3.\ Seeding methods.} In~\cite{AB09} the $k$-means++ randomized seeding algorithm by Arthur and Vassilvitskii~\cite{AV07} is generalized to $\mu$-similar Bregman divergences. 
Ackermann and Bl\"omer show that for a $\mu$-similar Bregman divergence this generalization, called Bregman++,  yields a 
$\O\bigl(\mu^{-2}\log(k)\bigr)$-approximation for \sbe. In~\cite{AB10} Ackermann and Bl\"omer generalize the result by Ostrovsky et a.~\cite{ORSS06} on adaptive sampling for $\epsilon$-separable instances to Bregman divergences.

Nock et al.~\cite{NLK08} generalize $k$-means++ to certain symmetrized versions of Bregman divergences $\dphi$, called \emph{mixed} Bregman divergences. They prove approximation factors of the form $\O\bigl(\rho_{\psi}\log k\bigr)$, where $\rho_{\psi}$ is some parameter depending on $\dphi$, that roughly measures how much $\dphi$ violates the triangle inequality.  Note, however, that the mixed Bregman divergences introduced by Nock et al.\ are not proper Bregman divergences.

\bibliographystyle{splncs03}
\bibliography{references}

\end{document}